\documentclass[a4paper,11pt]{article}
\usepackage[utf8x]{inputenc}

\usepackage{color}
\usepackage{a4wide}
\usepackage[fleqn]{amsmath}
\usepackage{amssymb}
\usepackage{amsthm}
\usepackage{shuffle}

\usepackage[shortlabels]{enumitem}
\setenumerate{itemsep=0.005ex}

 \usepackage{hyphenat}
\usepackage[boxed,algosection]{algorithm2e}
\SetAlCapSkip{1em}

\usepackage{tikz}
\usepackage{enumitem}
\usepackage{setspace}
\usetikzlibrary{arrows,automata,decorations.pathreplacing}

\theoremstyle{plain}
\newtheorem{theorem}{Theorem}
\newtheorem{lemma}[theorem]{Lemma}

\newtheorem{proposition}[theorem]{Proposition}
\theoremstyle{definition}
\newtheorem{definition}[theorem]{Definition}

\newcommand{\emptyword}[0]{\lambda}

\newcommand{\tuple}[1]{\langle #1 \rangle}

\newcommand{\limrun}[0]{\infty}

\newcommand{\trans}[1]{\mathchoice{\xrightarrow{#1}}{\xrightarrow{\smash{\lower1pt\hbox{$\scriptstyle #1$}}}}{\text{Error}}{\text{Error}}}

\newcommand{\alocc}[2]{|\!|#1|\!|_{#2}}
\newcommand{\occ}[2]{|#1|_{#2}}

\title{Finite-state independence and normal sequences}

\author{\begin{tabular}{lcr}
  Nicol\'as \'Alvarez \hspace*{1cm} & 
  Ver\'onica Becher & \hspace*{1cm}
  Olivier Carton 
\end{tabular}}

\begin{document}

\maketitle

\begin{abstract}
  We consider the previously defined notion of finite-state independence
  and we focus specifically on normal words.  We characterize finite-state
  independence of normal words in three different ways, using three
  different kinds of asynchronous deterministic finite automata with two
  input tapes containing infinite words.  Based on one of the
  characterizations we give an algorithm to construct a pair of
  finite-state independent normal words.
\end{abstract}
\bigskip
\bigskip

\section{Introduction and statement of results}

As defined by \'Emile Borel~\cite{Borel1909}, 
for an alphabet with at least two symbols, an infinite word $x$ is
{\em normal}  if all blocks of symbols of the same length occur in $x$ with the
same limiting frequency.   The most famous normal word was given by
Champernowne in~\cite{Champernowne1933},
\begin{displaymath}
01234567891011121314151617181920212223...
\end{displaymath}
Borel showed that almost all words are normal.  In~\cite{BCH2016} we
introduced the notion of \emph{finite-state independence} for pairs of
infinite words and we showed that almost all pairs of normal words are
finite-state independent.

In this work we characterize the notion finite-state-independence
specifically for normal words, in terms of computations in deterministic
asynchronous finite automata with two input tapes.  We give three
characterizations.

For the first characterization we consider the notion of \emph{fairness} of
a run in a given finite automaton for a given pair of input words.  A run
is fair if the frequency of each state is determined by the stationary
distribution associated with the automaton, hence not determined by the
input words.  This notion of fairness can also be phrased in terms of
frequencies of edges leaving each state.

The second characterization considers \emph{selectors}, which are finite
automata with two input tapes and one output tape such that the symbols in
the output tape are obtained by a selection of the symbols in the first
input tape, while the symbols in the second input tape act as a
consultative oracle.  We require that the selector be oblivious which means
that whether a symbol is selected or not does not depend on its value.
This characterization of finite-state independence of normal words extends
Agafonov's~\cite{Agafonov68} characterization of normality based on
selection by finite automata.

The third characterization considers \emph{shufflers}, which are finite
automata with two input tapes and one output tape such that, after the run,
the output tape contains all the symbols from the two normal words but
shuffled.  The output intercalates symbols from each of the input words,
preserving the order in which they appear in the input words.

We can now state the first theorem.

\begin{theorem}[Characterization Theorem] \label{thm:charac}
  Let $x$ and~$y$ be two normal words respectively on the alphabets $A$
  and~$B$. The following statements are equivalent.
  \begin{enumerate} \itemsep0cm
  \item \label{sta:independence} The words $x$ and $y$ are finite-state
    independent.
  \item \label{sta:fairness} For every deterministic two-tapes finite
    automaton~$\mathcal{A}$, the run on $x$ and~$y$ in~$\mathcal{A}$ is
    fair.
  \item \label{sta:selector} For every oblivious selector~$\mathcal{S}$,
    the results $\mathcal{S}(x,y)$ and~$\mathcal{S}(y,x)$ are also normal.
  \end{enumerate}
  Furthermore, if alphabets $A$ and~$B$ are equal, the following statement
  is also equivalent.
  \begin{enumerate}[resume]
  \item \label{sta:shuffler} For every shuffler~$\mathcal{S}$, the result
    $\mathcal{S}(x,y)$ is also normal.
  \end{enumerate}
\end{theorem}

Based on the characterization of finite-state independence of normal words
in terms of shufflers given in Theorem~\ref{thm:charac}, we obtain the
following.

\begin{theorem} \label{thm:algorithm}
  For every alphabet $A$, there is an algorithm that computes a pair of
  finite-state independent normal words.
\end{theorem}

The proof exhibits an algorithm that outputs a pair of finite-state
independent normal words $(x,y)$ by outputting, at each step, one new
symbol extending either the currently computed prefix of~$x$ or the
currently computed prefix of $y$.  Unfortunately, the computational
complexity of this algorithm is doubly exponential, which means that to
obtain the $n$-th symbol of the pair of finite-state independent normal
words the algorithm performs a number of operations that is doubly
exponential in~$n$.  Our construction of a pair of finite-state independent
normal words has some similarity with the construction of sequences
representing the fractional expansion of absolutely normal numbers (a
number is absolutely normal if its fractional expansion in each integer
base is a normal word).  Our algorithm here has some similarity with
Turing's algorithm for computing absolutely normal
numbers~\cite{turing,BFP2007}, which also has doubly exponential
computational complexity.

The paper is organized as follows.  In Section~\ref{section:primary} we
present the primary definitions of finite automata, normality and
finite-state independence.  We devote Section~\ref{sec:def} to the notions
of fairness, selecting and shuffling.  In Section~\ref{sec:proof} we give
the proof of Theorem~\ref{thm:charac} (Characterization Theorem).
Section~\ref{sec:algorithm} is devoted to Theorem~\ref{thm:algorithm},
which gives the announced algorithm to compute a pair of finite-state
independent normal words.  Finally in section \ref{sec:conclusion} we
report some open problems.

\section{Primary definitions}\label{section:primary}

Let $A$ be finite set of symbols, that we refer as the alphabet.  We write
$A^\omega$ for the set of all infinite words in alphabet $A$, $A^*$ for the
set of all finite words, $A^{\leq k}$ for the set of all words of length up
to $k$, and $A^k$ for the set of words of length exactly $k$.  The length
of a finite word $w$ is denoted by $|w|$.  The empty word is denoted
by~$\emptyword$.

\subsection{Normality}

We start with  some notation.
The positions of finite and infinite words are numbered starting at~$1$.  To
denote the symbol at position~$i$ of a word $w$ we write $w[i]$ and to
denote the substring of $w$ from position~$i$ to~$j$ we write $w[i..j]$.

\begin{definition}
  For $w$ and $u$ two words, the number $\occ{w}{u}$ of \emph{occurrences}
  of~$u$ in~$w$ and the number $\alocc{w}{u}$ of \emph{aligned occurrences}
  of~$u$ in~$w$ are respectively given by
  \begin{align*}
    \occ{w}{u}  & =|\{ i : w[i..i+|u|-1] = u \}|, \\
    \alocc{w}{u} & =|\{ i : w[i..i+|u|-1] = u \text{ and } i = 1 \mod |u|\}|.
  \end{align*}
\end{definition}
For example,  $\occ{aaaaa}{aa} = 4$ and
$\alocc{aaaaa}{aa} = 2$.
Notice that the definition of aligned occurrences has the condition
$i = 1 \mod |u|$ instead of $i = 0 \mod |u|$, because the positions are
numbered starting at~$1$.  Of course, when a word $u$ is just a symbol,
$\occ{w}{u}$ and $\alocc{w}{u}$ coincide.  Counting aligned occurrences of a
word of length $r$ over alphabet~$A$ is exactly the same as counting
occurrences of the corresponding symbol over alphabet~$A^r$. To be precise,
consider alphabet $A$, a length~$r$, and an alphabet~$B$ with $|A|^r$
symbols.  The set of words of length~$r$ over alphabet~$A$ and the set $B$
are isomorphic, as witnessed by the isomorphism $\pi: A^r \to B$ induced by
the lexicographic order in the respective sets.  Thus, for any $w \in A^*$
such that $|w|$ is a multiple of~$r$, $\pi(w)$ has length $|w|/r$ and
$\pi(u)$ has length $1$, as it is just a symbol in~$B$.  Then, for any
$u\in A^r$, $\alocc{w}{u} = \occ{\pi(w)}{\pi(u)}$.

We now present the definition of Borel normality~\cite{Borel1909} directly
on infinite words.    An infinite word $x$ is \emph{simply normal}
to word length~$\ell$ if, for every $u \in A^\ell$,
\begin{displaymath}
  \lim_{n \to \infty} \frac{\alocc{x[1..(n\ell)]}{u}}{n} = |A|^{-\ell}.
\end{displaymath}
An infinite word $x$ is \emph{normal} if it is simply normal to every word
length.  There are several other equivalent formulations of normality, they
can be read from~\cite{BC2018,Bugeaud12,KuiNie74}.

\subsection{Automata}

\begin{figure}[htbp]
  \begin{center}
    \begin{tikzpicture}[scale=1.4]
    \node (state) at (-1,2.5) [shape=rectangle,draw,rounded corners=1mm,
                            inner sep=10] {$Q$};
    \begin{scope}
      \draw (0.5,3.5)  -- (4.2,3.5);
      \draw[dotted] (4.2,3.5) -- (4.5,3.5);
      \draw (0.5,4)    -- (4.3,4);
      \draw[dotted] (4.3,4) -- (4.5,4);
      \foreach \x in {0.5,1,1.5,2,2.5,3,3.5,4}
        \draw (\x,3.5) -- (\x,4);
      \draw[very thick] (2,3.5) rectangle (2.5,4);
      \draw[->,>=latex,rounded corners] (state) -| ++(1,0.75) -| (2.25,3.5);
      \foreach \x/\xtext in {0.75/1,1.25/2,1.75/3,2.25/4,2.75/5,3.25/6,3.75/7}
        \node at (\x,3.7) {$a_{\xtext}$};
    \end{scope}
    \begin{scope}[shift={(0,-1)}]
      \draw (0.5,3.5)  -- (4.2,3.5);
      \draw[dotted] (4.2,3.5) -- (4.5,3.5);
      \draw (0.5,4)    -- (4.3,4);
      \draw[dotted] (4.3,4) -- (4.5,4);
      \foreach \x in {0.5,1,1.5,2,2.5,3,3.5,4}
        \draw (\x,3.5) -- (\x,4);
      \draw[very thick] (1,3.5) rectangle (1.5,4);
      \draw[->,>=latex,rounded corners] (state) -| ++(1,-0.25) -| (1.25,3.5);
      \foreach \x/\xtext in {0.75/1,1.25/2,1.75/3,2.25/4,2.75/5,3.25/6,3.75/7}
        \node at (\x,3.7) {$b_{\xtext}$};
    \end{scope}
    \begin{scope}[shift={(0.5,1)}]
      \draw (0,0.5)  -- (4.3,0.5);
      \draw[dotted] (4.3,0.5) -- (4.5,0.5);
      \draw (0,0)    -- (4.2,0);
      \draw[dotted] (4.2,0) -- (4.5,0);
      \foreach \x in {0,0.5,1,1.5,2,2.5,3,3.5,4}
        \draw (\x,0) -- (\x,0.5);
      \draw[very thick] (3,0) rectangle (3.5,0.5);
      \draw[->,>=latex,rounded corners] (state) -| ++(1,-0.75) -| (3.25,0.5);
      \foreach \x/\xtext in {0.25/1,0.75/2,1.25/3,1.75/4,2.25/5,2.75/6,3.25/7,3.75/8}
        \node at (\x,0.2) {$c_{\xtext}$};
    \end{scope}
    \end{tikzpicture}
  \end{center}
  \caption{Working principle of a $3$-automaton.}
  \label{fig:principle}
\end{figure}
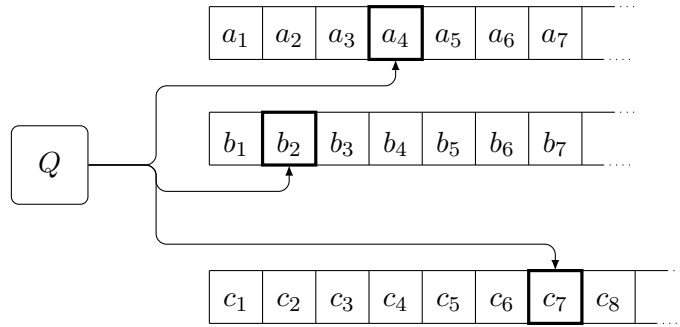

In this work we consider asynchronous finite automata running on a tuple of
infinite words with no accepting condition.  A thorough presentation of
these automata is in the books~\cite{PerrinPin04,Sakarovitch09}.

We consider $k$-tape automata, also known as $k$-tape transducers.  In the
rest of the paper, we use names such as \emph{compressors},
\emph{selectors} \emph{shufflers} or \emph{splitters} for some subclasses
of these automata to emphasize their use.  To simplify the presentation, we
assume here that the same alphabet~$A$ for the $k$ tapes.  A
\emph{$k$-automaton} is a tuple ${\mathcal A} = \tuple{Q,A,\delta,I}$,
where $Q$ is the finite state set, $A$ is the alphabet, $\delta$ is the
transition relation, $I$ the set of initial states.  The set of transition
relations is a finite subset of $Q \times (A^*)^k \times Q$.  A transition
is thus a tuple $\tuple{p,u_1,\ldots,u_k,q}$ where $p$ is its
\emph{starting state}, $\tuple{u_1,\ldots,u_k}$ is its \emph{label} and $q$
is its \emph{ending state} A transition is written
$p \trans{u_1,\ldots,u_k} q$.  As usual, two transitions are
\emph{consecutive} if the ending state of the first one is the starting
state of the second one.  A finite \emph{run} is a finite sequence of
consecutive transitions
\begin{displaymath}
  q_0 \trans{u_{1,1},\ldots,u_{k,1}} q_1 
      \trans{u_{1,2},\ldots,u_{k,2}} q_2 \cdots q_{n-1} 
      \trans{u_{1,n},\ldots,u_{k,n}} q_n.
\end{displaymath}
The \emph{label} of the run is the component-wise concatenation of the
labels of the transitions.  More precisely, it is the tuple
$\tuple{v_1,\ldots,v_k}$ where each $v_j$ for $1 \le j \le k$ is equal to
$u_{j,1}u_{j,2}\cdots u_{j,n}$.  Such a run is written shortly
as $q_0 \trans{v_1,\ldots,v_k} q_n$.
An infinite \emph{run}  is an infinite sequence of consecutive transitions
\begin{displaymath}
  q_0 \trans{u_{1,1},\ldots,u_{k,1}} q_1 
      \trans{u_{1,2},\ldots,u_{k,2}} q_2 
      \trans{u_{1,3},\ldots,u_{k,3}} q_3 \cdots 
\end{displaymath}
As for the finite case, the \emph{label} of the infinite run is the component-wise concatenation of the
labels of the transitions.  More precisely, it is the tuple
$\tuple{x_1,\ldots,x_k}$ where each $x_j$ for $1 \le j \le k$ is equal to
$u_{j,1}u_{j,2}u_{j,3}\cdots$.  Note that some label $x_j$ might be finite
although the run is infinite since some transitions may have empty
labels. The run is accepting if its first state~$q_0$ is initial and each
word~$x_j$ is infinite.  Such an accepting run is written shortly
$q_0 \trans{x_1,\ldots,x_k} \limrun$.  The tuple $\tuple{x_1,\ldots,x_k}$
is accepted if there exists at least one accepting run with label
$\tuple{x_1,\ldots,x_k}$.  Notice that there is no constraint on the states
occurring infinitely often in an accepting run.

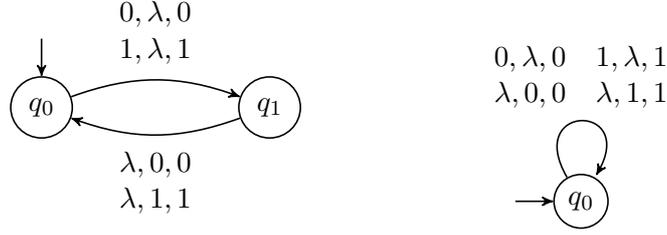
\begin{figure}
  \begin{center}
    \begin{tikzpicture}[->,>=stealth',initial text=,semithick,auto,inner sep=4pt]
      \tikzstyle{every state}=[minimum size=0.5]
      \node[state,initial above] (q0) at (0,0)  {$q_0$};
      \node[state]  (q1) at (3,0) {$q_1$};
      \path (q0) edge[bend left=20] node {$\begin{array}{c} 
                                             0,\emptyword,0 \\
                                             1,\emptyword,1
                                          \end{array}$} (q1);
      \path (q1) edge[bend left=20] node {$\begin{array}{c} 
                                             \emptyword,0,0 \\
                                             \emptyword,1,1
                                          \end{array}$} (q0);
    \end{tikzpicture}
\hspace*{2cm}   
\begin{tikzpicture}[->,>=stealth',initial text=,semithick,auto,inner sep=3pt]
      \tikzstyle{every state}=[minimum size=0.4]
      \node[state,initial left] (q0) at (0,0)  {$q_0$};
      \path (q0) edge[out=120,in=60,loop] node {$\begin{array}{c} 
                                       0,\emptyword,0 \quad 1,\emptyword,1 \\
                                       \emptyword,0,0 \quad \emptyword,1,1
                                                \end{array}$} ();
    \end{tikzpicture}
  \end{center}
  \caption{A $2$-deterministic $3$-automaton (left) 
\label{fig:join} 
 and    a non-deterministic $3$-automaton (right)}
  \label{fig:shuffle}
\end{figure}

In this work we consider only deterministic $k$-automata whose transition
function is determined by a subset of the $k$ tapes.  We say that a
$k$-automaton is \emph{$\ell$-deterministic}, with $1 \le \ell \le k$, if
the following two conditions are fulfilled:
\begin{enumerate} \itemsep0cm
\item the set $I$ of initial states is a singleton set;
\item for each state~$p$, there is an integer~$1 \le i \le \ell$ such that
  for each transition $p \trans{u_1,\ldots,u_k} q$ starting from~$p$, $u_i$
  is a symbol and $u_1,\ldots, u_{i-1},u_{i+1},\ldots,u_{\ell}$ are empty.
  Furthermore if $p \trans{u_1,\ldots,u_k} q$ and
  $p \trans{u'_1,\ldots,u'_k} q'$ are two transitions starting from~$p$,
  then $u_i \neq u'_i$.
\end{enumerate}

The $\ell$-deterministic automaton is called \emph{$\ell$-complete} if for
each state~$p$ and each symbol~$a$, there is an integer~$i$ (depending only
on~$p$) and a transition $p \trans{u_1,\ldots,u_k} q$ starting from~$p$ such
that $1 \le i \le \ell$ and $u_i = a$.
The $\ell$-determinism guarantees that for each tuple
$\tuple{x_1,\ldots,x_{\ell}}$ of infinite words, there exists at most one
run such that the first $\ell$ components of its label are
$\tuple{x_1,\ldots,x_{\ell}}$.  Even if the automaton is $\ell$-complete,
this run might be not accepting since one of its labels might be finite.

The $3$-automaton at the left of Figure~\ref{fig:join} accepts a triple
$\tuple{x,y,z}$ of infinite words over the alphabet $\{ 0, 1 \}$ whenever
$z$ is the join of $x$ and~$y$; recall that the join of two infinite words
$x = a_1a_2a_3\cdots$ and $y = b_1b_2b_3\cdots$ is the infinite word
$z = a_1b_1a_2b_2a_3\cdots$.  This automaton is $2$-deterministic.
The $3$-automaton pictured at the right of Figure~\ref{fig:shuffle} accepts
a triple $\tuple{x,y,z}$ of infinite words over the alphabet $\{ 0, 1 \}$
whenever $z$ is a shuffle of the symbols in $x$ and~$y$. This automaton is
not $2$-deterministic.  Indeed the first condition on transitions is not
fulfilled by the two transitions $q_0 \trans{0,\emptyword,0} q_0$ and
$q_0 \trans{\emptyword,0,0} q_0$.

Let $\mathcal{A}$ be an $\ell$-deterministic $k$-automaton.  
For each tuple
$\tuple{x_1,\ldots,x_{\ell}}$ of infinite words, there exists at most one
tuple $\tuple{y_{\ell+1},\ldots,y_k}$ of infinite words such that the
$k$-tuple \linebreak
$\tuple{x_1,\ldots,x_{\ell},y_{\ell+1},\ldots,y_k}$ is accepted
by~$\mathcal{A}$.  The automaton~$\mathcal{A}$ realizes then a partial
function from $(A^\omega)^{\ell}$ to $(A^\omega)^{k-\ell}$ and the tuple
$\tuple{y_{\ell+1},\ldots,y_k}$ is denoted by
$\mathcal{A}(x_1,\ldots,x_{\ell})$.  The $1$-deterministic $2$-automata are
also called sequential transducers in the literature.  When a $k$-automaton
is $\ell$-deterministic, each transition is written
\begin{displaymath}
  p \trans{u_1,\ldots,u_{\ell}|v_{\ell+1},\ldots, v_k} q
\end{displaymath}
to emphasize that the first $\ell$ tapes are input tapes and that
the $k-\ell$ remaining ones are output tapes.

Let $\mathcal{A}$ be a $1$-deterministic $2$-automaton.  We say that
$\mathcal{A}$ is a \emph{compressor} if the (partial) function
$x \mapsto \mathcal{A}(x)$ which maps $x$ to the output~$\mathcal{A}(x)$ is
one-to-one.
The compression ratio of an infinite word~$x$ for~$\mathcal{A}$ is given by
the unique accepting run
$q_0 \trans{u_1|v_1} q_1 \trans{u_2|v_2} q_2 \trans{u_3|v_3} q_3 \cdots$
where $x = u_1u_2u_3\cdots$ as
\begin{displaymath}
 \rho_{\mathcal{A}}(x) = 
    \liminf_{n \to \infty} 
    \frac{|v_1v_2\cdots v_n|}{|u_1u_2\cdots v_n|}.
\end{displaymath} 
This compression ratio for a given automaton~$\mathcal{A}$ can have
any non-negative real value.  In particular, it can be greater than~$1$.
An infinite word $x$ is \emph{compressible} by a $1$-deterministic
$2$-automaton~$\mathcal{A}$ if $\rho_{\mathcal{A}}(x) < 1$. 
 The \emph{compression ratio} of a given word~$x, $  $\rho(x)$,
 is the infimum of the compression ratios achievable by all one-to-one $1$-deterministic
$2$-automata, namely,
\begin{displaymath}
  \rho(x) = \inf\{\rho_{\mathcal{A}}(x) :  
   \text{$\mathcal{A}$ is a one-to-one $1$-deterministic $2$-automaton}\}
\end{displaymath}
For every infinite word $x$, $\rho(x)$ is less than or equal to~$1$,
because there exists a compressor~$\mathcal{A}_0$ which copies each symbol
of the input to the output, so $\rho_{\mathcal{A}_0}(x)$ is equal to~$1$.
The compression ratio of the word $x = 0^{\omega}$ is $\rho(x) = 0$ because
for each positive real number~$\varepsilon$ there exists a
compressor~$\mathcal{A}$ such that $\rho_{\mathcal{A}}(x) < \varepsilon$.
Notice that in this case the compression ratio equal to~$0$ is not
achievable by any compressor~$\mathcal{A}$.  It follows from the results
in~\cite{Schnorr71,Dai04} that the words $x$ with compression ratio
$\rho(x)$ equal to~$1$ are the exactly the normal words.  A direct proof of
this result appears in~\cite[Characterization Theorem]{BC2018}.

\subsection{Finite-state independence}

Roughly, two infinite words, possibly over different alphabets, are
finite-state independent if none of them helps to compress the other using
$3$-automata.  In our setting, a \emph{compressor} is a $2$-deterministic
$3$-automata~$\mathcal{A}$ such that for any fixed infinite word~$y$, the
function $x \mapsto \mathcal{A}(x,y)$ which maps $x$ to the
output~$\mathcal{A}(x,y)$ is one-to-one.  This guarantees that if $y$ is
known, $x$ can be recovered from $\mathcal{A}(x,y)$.  Note that we do not
require that the function $(x,y) \mapsto \mathcal{A}(x,y)$ be one-to-one,
which would be a much stronger assumption.  For example, the
$2$-deterministic $3$-automaton~$\mathcal{A}$ which maps the infinite
words~$x$ and~$y$ to the infinite word~$z$ satisfying
$z[i] = x[i] + y[i] \mod |A|$ for each $i \ge 1$ is, indeed a compressor
but the function $(x,y) \mapsto \mathcal{C}(x,y)$ is not one-to-one.

\begin{definition}[\cite{BCH2016}]
  Let $\mathcal{A}$ be a compressor.  For simplicity in the presentation we
  assume just one alphabet.  However, it is possible to have three
  different alphabets, one for each input tape and one for the output tape.
  The \emph{conditional compression ratio} of an infinite word~$x$ with
  respect to~$y$ in~$\mathcal{A}$ is given by the unique accepting run
  \begin{displaymath}
    q_0 \trans{u_1,v_1|w_1} q_1 
        \trans{u_2,v_2|w_2} q_2 
        \trans{u_3,v_3|w_3} q_3 \cdots
  \end{displaymath}
  such that $x = u_1u_2u_3\cdots$ and $y = v_1v_2v_3\ldots$ as
  \begin{displaymath}
     \rho_{\mathcal{A}}(x/y) = 
           \liminf_{n \to \infty} \frac{|w_1w_2w_3\cdots|}{|u_1u_2u_3\cdots|}.
  \end{displaymath} 
  In case the input tape and the output tape have respective alphabets $A$
  and $B$ of different sizes, the formula above should be multiplied by
  $\log |A|/\log |B|$.  Notice that the number of symbols read from~$y$,
  namely $|v_1v_2v_3\cdots|$, is not taken into account in the value
  of~$\rho_{\mathcal{A}}(x/y)$.  

The \emph{conditional compression ratio}
  of an infinite word $x$ given an infinite word~$y$, $\rho(x/y)$, is the
  infimum of the compression ratios $\rho_{\mathcal{A}}(x/y)$ of all
  compressors~$\mathcal{A}$ with input~$x$ and oracle~$y$.
\end{definition}

\begin{definition}[\cite{BCH2016}]
  Two infinite words $x$ and~$y$, possibly over different alphabets, are
  \emph{finite-state independent} if $\rho(x/y) = \rho(x)$,
  $\rho(y/x) = \rho(y)$ and the compression ratios of $x$ and~$y$ are
  non-zero.
\end{definition}

Notice that the compression ratios of $x$ and~$y$ should not be zero.  This
means that a word~$x$ such that $\rho(x) = 0$ is finite-state independent
of no word.  Without this requirement, two words $x$ and~$y$ such that
$\rho(x) = \rho(y) = 0$ would be finite-state independent.  In particular,
each word~$x$ with $\rho(x) = 0$ would be finite-state independent of
itself.  From the definition of finite-state independence follows that, if
the infinite words $x$ and~$y$ are finite-state independent, each suffix
of~$x$ is finite-state independent of each suffix of~$y$.

Finite-state independence for a pair of normal words differs from the
classical notion of normality for dimension~$2$ also known as \emph{joint
  normality} \cite{KuiNie74}.  When two normal words are finite-state
independent then they are also jointly normal, but the reverse implication
fails.  A witness for this appears in~\cite{BCH2016} with two normal words
$x$ and $y$ such that $x$ is identical to the intercalation of the symbols
of $x$ and $y$ (in our construction the sequence $x$ satisfies that for
every position $n$, $x[n]=x[2n]$).  Thus, from the normality of $x$ follows
that $x$ and $y$ are jointly normal.  However, given $x$ we can obtain $y$
as the subsequence of $x$ in the odd positions, hence $x$ and $y$ are not
finite-state independent.  This already suggests that the concept of
finite-state independence can not be obtained with synchronous automata.

Although finite-state independence of normal words is more demanding than
joint normality, it still holds that almost all pairs of normal words are
finite-state independent. This is proved in~\cite[Theorem 5.1]{BCH2016}.

\section{Fairness, selecting and shuffling}\label{sec:def}

\subsection{Fairness}

We use the terminology of Markov chains
for strongly connected components of an automaton.  A strongly connected
component of an automaton is called \emph{recurrent} if any state reachable
from it is still in it.  It is called \emph{transient} otherwise.  By
extension, a state is called \emph{recurrent} (respectively,
\emph{transient}) whenever it belongs to a recurrent (respectively,
transient) strongly connected component.  Let $\mathcal{A}$ be a
$2$-deterministic $2$-automaton and let $x$ and $y$ be two infinite words,
possibly over different alphabets.  Let $\gamma$ be the run
of~$\mathcal{A}$ on $x$ and~$y$
\begin{displaymath}
  q_0 \trans{\bar{a}_1,\bar{b}_1} 
  q_1 \trans{\bar{a}_2,\bar{b}_2} 
  q_2 \trans{\bar{a}_3,\bar{b}_3} q_3 \cdots
\end{displaymath}
where each $\bar{a}_i$ and each $\bar{b}_i$ is either a symbol or the empty
word and each $q_{i-1} \trans{\bar{a}_i,\bar{b}_i} q_i$ is a transition
of~$\mathcal{A}$.  With a slight abuse of notation let
$\occ{\gamma[1..n]}{q}$ denote the number of occurrences of the state~$q$
in the first $n$ states of~$\gamma$.  More precisely, this is the
cardinality of the set
\begin{displaymath}
  \{ i:  0 \le i \le n-1,\;\; q_i = q \}.
\end{displaymath}
Similarly, for each transition $\tau = p \trans{\bar{a},\bar{b}} q$ let
$\occ{\gamma[1..n]}{\tau}$ denote the number of occurrences of~$\tau$ in
the first $n$ transitions of~$\gamma$.  More precisely, this is the
cardinality of the set
\begin{displaymath}
  \{i :  1 \le i \le n,\;\; q_{i-1} \trans{\bar{a}_i,\bar{b}_i} q_i = \tau \}.
\end{displaymath}

We first introduce the notion of fairness for states.  It is based on
on a notion of stationary distribution of an automaton which is now
defined.
We associate with $2$-deterministic and $2$-complete
$2$-automaton~$\mathcal{A}$ a Markov chain described by a stochastic
matrix~$M$.  Let $A$ and $B$ be the alphabets for the first and second tape
of~$\mathcal{A}$.  The state set of the Markov chain is the state set~$Q$
of~$\mathcal{A}$.  The dimension of the matrix~$M$ is thus the number~$|Q|$
of states and its rows and columns are indexed by element of~$Q$.  For two
states $p$ and~$q$, the $(p,q)$-entry of $M$ is the sum of the weights of
all transitions from~$p$ to~$q$ where the weights are as follows.  The
weight of a transition of the form $p \trans{a,\emptyword} q$ (respectively
$p \trans{\emptyword,b} q$) is $1/|A|$ (respectively $1/|B|$).

\begin{figure}[htbp]
  \begin{center}
    \begin{tikzpicture}[->,>=stealth',initial text=,semithick,auto,inner sep=3pt]
      \tikzstyle{every state}=[minimum size=0.4]
      \node[state,initial left] (q0) at (0,0)  {$q_1$};
      \node[state] (q1) at (2,0)  {$q_2$};
      \path (q0) edge[out=120,in=60,loop] node {$0,\emptyword$} ();
      \path (q0) edge[bend left=20] node {$1,\emptyword$} (q1);
      \path (q1) edge[bend left=20] node {$\begin{array}{c} 
                                            \emptyword,0 \\
                                            \emptyword,1
                                           \end{array}$} (q0);
    \end{tikzpicture}
  \end{center}
  \caption{A $2$-deterministic $2$-automaton}
  \label{fig:distrib}
\end{figure}
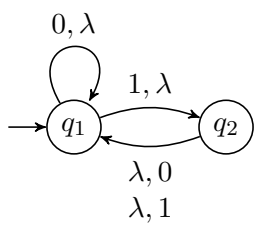

If the automaton~$\mathcal{A}$ is strongly connected then the Markov chain
is irreducible.  By~\cite[Theorem~1.5]{Senata06}, there exists a unique
stationary distribution, that is, a line vector~$\pi$ such that
$\pi M = \pi$ and $\sum_{q \in Q}{\pi(q)} = 1$.  By definition, this vector
is called the \emph{stationary distribution} of the
automaton~$\mathcal{A}$.  For example, the matrix of the associated Markov
chain for the $2$-automaton in Figure~\ref{fig:distrib} is the
$2 \times 2$-matrix~$M$ given by
\begin{displaymath}
 M = \left( 
     \begin{matrix} 
       \frac{1}{2} & \frac{1}{2} \\ 
       1 & 0 
     \end{matrix} 
     \right)
\end{displaymath}
and the stationary distribution is thus given by $\pi(q_1) = 2/3$ and
$\pi(q_2) = 1/3$.

If the automaton~$\mathcal{A}$ is not strongly connected, the stationary
distribution~$\pi$ of~$\mathcal{A}$ is defined as follows.  For each
transient state~$q$, $\pi(q)$ is equal to~$0$.  For any recurrent
state~$q$, $\pi(q) = \hat{\pi}(q)$ where $\hat{\pi}$ is the stationary
distribution of the strongly connected component of~$q$, considered as a
whole automaton.  This is well-defined because this stationary
distribution only depends on the edges of the automaton and not on its
initial and final states.

Let $\mathcal{A}$ be a $2$-deterministic $2$-automaton and let $x$ and $y$
be two infinite words, possibly over different alphabets.  Let
$\gamma = q_0 \trans{\bar{a}_1,\bar{b}_1} q_1 \trans{\bar{a}_2,\bar{b}_2}
q_2 \cdots$ be the run of~$\mathcal{A}$ on $x$ and~$y$.  This run is
called \emph{fair for states} if for any state~$q$ which occurs
in~$\gamma$,
\begin{displaymath}
  \lim_{n\to\infty}{\frac{\occ{\gamma[1..n]}{q}}{n}} = \pi(q).
\end{displaymath}
The run~$\gamma$ is called \emph{fair for edges} if for any pair of
transitions $\tau$ and~$\tau'$ starting from the state
\begin{displaymath}
  \lim_{n\to\infty}{\frac{\occ{\gamma[1..n]}{\tau}}{n}} = 
  \lim_{n\to\infty}{\frac{\occ{\gamma[1..n]}{\tau'}}{n}}.
\end{displaymath}

Let $\mathcal{A}$ be a $2$-deterministic $2$-automaton.  By analogy with
the line graph, the \emph{line automaton} of~$\mathcal{A}$ is the
automaton~$\hat{\mathcal{A}}$ whose states are the transitions
of~$\mathcal{A}$.  More formally, its state set is
$\hat{Q} = E \cup \{\tau_0\}$ where $E$ is the set of transitions
of~$\mathcal{A}$ and $\tau_0$ is a fresh element  (it does not belong to $E$) 
 being the initial state. Its set~$\hat{E}$ of transitions is given by
\begin{displaymath}
  \hat{E} =
  \{ \tau_0 \trans{\bar{a},\bar{b}} \tau :
  q_0 \in I, \tau = q_0 \trans{\bar{a},\bar{b}} q\}  \cup
  \{ \tau' \trans{\bar{a},\bar{b}} \tau :
  \tau' = r \trans{\bar{c},\bar{d}} p, \tau = p \trans{\bar{a},\bar{b}} q\}.
\end{displaymath}

There is a tight correspondence between runs in~$\mathcal{A}$ and runs
in~$\hat{\mathcal{A}}$.  To each run in~$\mathcal{A}$
\begin{displaymath}
  q_0 \trans{\bar{a}_1,\bar{b}_1} 
  q_1 \trans{\bar{a}_2,\bar{b}_2} 
  q_2 \trans{\bar{a}_3,\bar{b}_3} q_3 \cdots
\end{displaymath}
starting from the initial state~$q_0$ of~$\mathcal{A}$ corresponds the run
\begin{displaymath}
  \tau_0 \trans{\bar{a}_1,\bar{b}_1} 
  \tau_1 \trans{\bar{a}_2,\bar{b}_2} 
  \tau_2 \trans{\bar{a}_3,\bar{b}_3} \tau_3 \cdots
\end{displaymath}
where $\tau_0$ is the fresh initial state of~$\hat{\mathcal{A}}$ and
$\tau_i$ is the transition $q_{i-1} \trans{\bar{a}_i,\bar{b}_i} q_i$
for each $i \ge 1$.  Conversely, each run in~$\hat{\mathcal{A}}$ starting
from~$\tau_0$ comes from a run in~$\mathcal{A}$.  The following lemma
relates the stationary distribution of~$\hat{\mathcal{A}}$ with the
stationary distribution of~$\mathcal{A}$.
\begin{lemma} \label{lem:dist-line}
  The stationary distribution~$\hat{\pi}$ of~$\hat{\mathcal{A}}$ maps each
  transition $\tau = p \trans{\bar{a},\bar{b}} q$ to
  $\hat{\pi}(\tau) = \pi(p)/n_p$ where $\pi$ is the stationary distribution
  of~$\mathcal{A}$ and $n_p$ is the number of transitions starting from
  state~$p$ in~$\mathcal{A}$.
\end{lemma}
\begin{proof}
The proof of the Lemma \ref{lem:dist-line} is routine.
\end{proof}

Next we relate the fairness for states and the fairness
for edges: the two notions are equivalent as long as they
hold for all automata.  We start with an auxiliary lemma.

\begin{lemma} \label{lem:fairedgesscc}
  A run which is fair for edges ends in a recurrent strongly connected
  component.
\end{lemma}
\begin{proof}
  Let $P$ be the subset of states
  $\{ q : \liminf_{n\to\infty}{\frac{\occ{\gamma[1..k_n]}{q}}{n}} > 0\}$.
  The set~$P$ cannot be empty because there a finitely many states.
  The hypothesis implies that every state reachable from a state in~$P$ is
  also in~$P$.  Since a recurrent state is reachable from any state,
  $\gamma$ reaches a recurrent state as it was claimed. 
\end{proof}

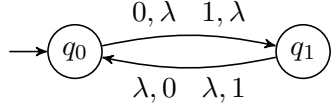
\begin{figure}[htbp]
  \begin{center}
    \begin{tikzpicture}[->,>=stealth',initial text=,semithick,auto,inner sep=3pt]
      \tikzstyle{every state}=[minimum size=0.4]
      \node[state,initial left] (q0) at (0,0)  {$q_0$};
      \node[state] (q1) at (3,0)  {$q_1$};
      \path (q0) edge[bend left=12] node {$0,\emptyword \;\;\; 1,\emptyword$} (q1);
      \path (q1) edge[bend left=12] node {$\emptyword,0 \;\;\; \emptyword,1$} (q0);
    \end{tikzpicture}
  \end{center}
\vspace*{-0.8cm}
  \caption{Another $2$-deterministic $2$-automaton}
  \label{fig:non-fairness}
\end{figure}

Consider the $2$-deterministic $2$-automaton~$\mathcal{A}$ pictured in
Figure~\ref{fig:non-fairness} and the infinite words $x = y = 0^\omega$.
The run on $x$ an~$y$ in~$\mathcal{A}$ is fair for states because both
states $q_0$ and~$q_1$ have frequency $1/2$ but it is not fair for edges
because the transition $q_0 \trans{1,\emptyword} q_1$ is never used.

\begin{proposition} \label{pro:fairness}
  Let $x$ and $y$ be two infinite words.  The run on $x$ and~$y$ is fair
  for states in any $2$-deterministic $2$-automaton if and only if it is
  fair for edges in any $2$-deterministic $2$-automaton.
\end{proposition}

\begin{proof}
  We first show that fairness for states implies fairness for edges.  Let
  $\mathcal{A}$ be a $2$-deterministic $2$-automaton.  This implication
  follows from the hypothesis applied to the line
  automaton~$\hat{\mathcal{A}}$ and Lemma~\ref{lem:dist-line}.
  
  We now prove that fairness for states and fairness for edges coincide.
  Let $\mathcal{A}$ be a $2$-deterministic $2$-automaton and let $\gamma$
  be the run on $x$ and~$y$ in~$\mathcal{A}$.  By
  Lemma~\ref{lem:fairedgesscc} the run~$\gamma$ visits a recurrent state
  of~$\mathcal{A}$.  Therefore, we now assume that $\mathcal{A}$ is
  strongly connected.  To prove the statement about frequencies of states,
  it is sufficient to show that for each increasing sequence of integers
  $(k_n)_{n\ge0}$ such that
  $\lim_{n\to\infty}{\occ{\gamma[1..k_n]}{q}/k_n}$ exists, this limit is
  equal to~$\pi(q)$.  Let $(k_n)_{n\ge0}$ be such a sequence.  Replace
  $(k_n)_{n\ge0}$ by one of its sub-sequences so that
  $\lim_{n\to\infty}{\occ{\gamma[1..k_n]}{q}/k_n}$ exists for each
  state~$q$.  It has already been shown in the previous paragraph that
  these limits cannot be $0$.

  We introduce two sequences $(v_n)_{n\ge0}$ and $(v'_n)_{n\ge0}$ of line
  vectors and a sequence $(M_n)_{n\ge0}$ of matrices. For each state~$q$,
  the $q$-entries of the vectors $v_n$ and~$v'_n$ are given by
  \[
   v_n(q) = \occ{\gamma[1..k_n]}{q}/k_n\] 
  and
  \[ v'_n(q) = \occ{\gamma[2..k_n+1]}{q}/k_n.
  \]
 For each pair of states $p$
  and~$q$, the $(p,q)$-entry of $M_n$ is the sum over all
  transitions~$\tau$ from~$p$ to~$q$ of the ratio
  $\occ{\gamma[1..k_n]}{\tau}/\occ{\gamma[1..k_n]}{p}$.  A routine check
  yields that $v_nM_n = v'_n$ holds for each integer $n \ge 1$.  Both
  sequences $(v_n)_{n\ge0}$ and $(v'_n)_{n\ge0}$ converge to the same line
  vector~$v$ given by
  $v(q) = \lim_{n\to\infty}{\occ{\gamma[1..k_n]}{q}/k_n}$.  From the
  hypothesis, the sequence $(M_n)_{n\ge0}$ converges to the matrix~$M$ of
  the Markov chain associated with~$\mathcal{A}$.  Taking limits gives that
  $vM = v$.  By the uniqueness of the stationary distribution of~$M$,
  $v(q) = \pi(q)$ holds for each state~$q$.
\end{proof}

By Proposition~\ref{pro:fairness}, the two notions of fairness, fairness
for states and fairness for edges are equivalent.  This allows us to use
the notion of fairness without mentioning which one is meant.  The
following lemmas on fairness are used in the proof of
Theorem~\ref{thm:charac}.

Let $\mathcal{A}$ be $2$-deterministic $2$-automaton and let $k$ and $\ell$
be two positive integers.  We introduce a new automaton
$\mathcal{A}_{k,\ell}$.  Its state set is
$Q \times A^{\le k} \times \{\emptyword\} \cup Q \times A^k \times B^{\le
  \ell}$ and its transitions are defined as follows.
\begin{align*}
   (q,u,\emptyword) \trans{a,\emptyword} (q,ua,\emptyword) & \quad
    \text{if $|u| < k$} \\
   (q,u,v) \trans{\emptyword,b} (q,u,vb) & \quad
    \text{if $|u| =  k$ and $|v| < \ell$} \\
   (q,au',v) \trans{a',\emptyword} (q,u'a',v) & \quad
    \text{if $|u'| = k-1$, $|v| = \ell$ and 
          $q \trans{a,\emptyword} q'$ in $\mathcal{A}$} \\
   (q,u,bv') \trans{\emptyword,b'} (q,u,v'b') & \quad
    \text{if $|u| = k$, $|v'| = \ell-1$ and 
          $q \trans{\emptyword,b} q'$ in $\mathcal{A}$}
\end{align*}
Note that the states in
$Q \times A^{\le k} \times \{\emptyword\} \cup Q \times A^k \times
B^{<\ell}$
are obviously transient.  The purpose of these states is to gather the
first $k$ symbols of~$x$ and the first $\ell$ symbols of~$y$ to reach the
state $(q_0,u,v)$ where $q_0$ is the initial state of~$\mathcal{A}$ and $u$
and~$v$ are the prefixes of~$x$ and~$y$ of length $k$ and~$\ell$
respectively.

\begin{lemma}\label{lemma:1}
  If $\mathcal{A}$ is strongly connected, then the restriction of
  $\mathcal{A}_{k,\ell}$ to the set $Q \times A^k \times B^{\ell}$
  is also strongly connected.
\end{lemma}

\begin{proof}
  Let $(q,u,v)$ and $(q',u',v')$ be two states in
  $Q \times A^k \times B^{\ell}$.  There exist a word~$w$ in $A^* \cup B^*$
  and  states~$r$ of~$\mathcal{A}$ such that either $q \trans{uw,v} r$ or
  $q \trans{u,vw} r$ is a finite run in~$\mathcal{A}$.  By symmetry, it can
  be assumed that $q \trans{uw,v} r$ is a finite run in~$\mathcal{A}$.
  Since $\mathcal{A}$ is strongly connected, there exists a run
  $r \trans{u'',v''} q'$.  Then
  \begin{displaymath}
    (q,u,v) \trans{uwu''u',vv''v'} (q',u',v')
  \end{displaymath}
  is a run in~$\mathcal{A}_{k,\ell}$.
\end{proof}

\begin{lemma}\label{lemma:2}
  If $\mathcal{A}$ is strongly connected and $\pi$ is its stationary
  distribution, then the stationary distribution of~$\mathcal{A}_{k,\ell}$
  is given by $\pi(q,u,v) = \pi(q)/|A|^k|B|^{\ell}$ for each state 
  $(q,u,v) \in Q \times A^k \times B^{\ell}$.
\end{lemma}
\begin{proof}
  Le $M = (m_{p,q})$ be the $Q \times Q$-matrix of~$\mathcal{A}$.  Each
  entry $m_{p,q}$ is equal to either
  \begin{displaymath}
    |\{ a : p \trans{a,\emptyword} q \}|/|A| 
    \text{ or } 
    |\{ b : p \trans{\emptyword,b} q \}|/|B| 
  \end{displaymath}
  The vector $\pi$ is the unique vector satisfying $\pi M = \pi$ 
  and $\sum_{q\in Q}{\pi(q)} = 1$.  
  Let $(q,u,v)$ be a fixed state.  For each transition $p
  \trans{a,\emptyword} q$, there is a transition 
  $(p,au',v) \trans{a',\emptyword} (q,u,v)$ where $u = u'a'$ ($u'$ is the
  prefix of length $k-1$ of~$u$ and $a$ is its last symbol).
\end{proof}

\subsection{Selecting} \label{sec:selector}

We present the definition of a selector that we use to characterize
finite-state independence of normal words, to be given in
Theorem~\ref{thm:charac}.  Given a normal infinite word, the problem of
selection is how to select symbols from an infinite word so that the word
defined by the selected symbols satisfies a designated property.  An early
result of Wall~\cite{Wall49} shows that selecting the symbols of a normal
word in the positions given by an arithmetical progression yields again a
normal word.  Agafonov~\cite{Agafonov68} extended Wall's result and proved
that any selection by finite automata preserves normality (a complete proof
can be found in~\cite[Theorem 7.1]{BecherCartonHeiber15}).  The selections
admitted by Agafonov must be performed by an oblivious $1$-deterministic
$2$-automaton.  Oblivious means that the choice of selecting or not the
next symbol only depends on the current state and not on the next symbol.

Other forms of selection by finite-automata do not preserve normality. For
instance~\cite[Theorem 7.3]{BecherCartonHeiber15} shows that the
two-sided selection rule ``select symbols in between two zeroes'' from $x$,
does not preserve normality.  

In order to characterize finite-state independence we consider selection by
a finite automaton from an infinite word, conditioned to another infinite
word that can be used in the selection process as an oracle.

\begin{definition}
  A \emph{selector} is a $2$-deterministic $3$-automaton such that each of
  its transitions has one of the types $p \trans{a,\emptyword|a} q$ (type~I),
  $p \trans{a,\emptyword|\emptyword} q$ (type~II), or
  $p \trans{\emptyword,b|\emptyword} q$ (type~III) for two symbols
  $a,b \in A$.  It is \emph{oblivious} if all transitions starting at a
  given state have the same type.
\end{definition}

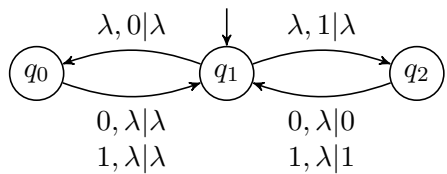
\begin{figure}[htbp]
    \begin{center}
    \begin{tikzpicture}[->,>=stealth',initial text=,semithick,auto,inner sep=3pt]
      \tikzstyle{every state}=[minimum size=0.4]
      \node[state] (q0) at (0,0)  {$q_0$};
      \node[state,initial above]  (q1) at (2.5,0) {$q_1$};
      \node[state]  (q2) at (5,0) {$q_2$};
      \path (q1) edge[bend right=20,above] node {$\emptyword,0|\emptyword$} (q0);
      \path (q1) edge[bend left=20] node {$\emptyword,1|\emptyword$} (q2);
      \path (q0) edge[bend right=20,below] node {$\begin{array}{c} 
                                            0,\emptyword|\emptyword \\
                                            1,\emptyword|\emptyword
                                           \end{array}$} (q1);
      \path (q2) edge[bend left=20] node {$\begin{array}{c} 
                                            0,\emptyword|0 \\
                                            1,\emptyword|1
                                           \end{array}$} (q1);
    \end{tikzpicture}
  \end{center}
\vspace*{-0.8cm}
  \caption{An oblivious selector}
  \label{fig:selector}
\end{figure}

A transition of type $p \trans{a,\emptyword|a} q$ (type~I) copies a
symbol from the first input~$x$ to the output tape.  A transition of the
types $p \trans{a,\emptyword|\emptyword} q$ (type~II) or
$p \trans{\emptyword,b|\emptyword} q$ (type~III) skips a symbol from either
the first input~$x$ or the second input~$y$. It follows then that the
output word~$z = \mathcal{S}(x,y)$ is obtained by selecting symbols
from~$x$.  This justifies the terminology.

Since a selector is $2$-deterministic, all transitions starting at a
given state either have type I and~II or have type~III.  When it is 
oblivious it is not possible anymore that two transitions starting at
the same state have types I and~II.  Whether or not a symbol is copied from 
the first input tape to the output tape only depends on the state and 
not on the symbol.

The automaton pictured in Figure~\ref{fig:selector} is an oblivious
selector.  It selects symbols from the first input~$x$ which are at
a position where there is a symbol~$1$ in the second input~$y$.

\subsection{Shuffling}

We present the definition of a shuffler we use to characterize finite-state
independence of normal words in Theorem~\ref{thm:charac}. A general
presentation of shufflers can be read in~\cite{Pin1989}.  An infinite
word~$z$ is the shuffle of $x$ and~$y$ if it can be factorized as
$z = u_1v_1u_2v_2u_3\cdots$ where the sequences of words $(u_i)_{i\ge1}$
and $(v_i)_{i\ge1}$ satisfy $x = u_1u_2u_3\cdots$ and
$y = v_1v_2v_3\cdots$.  We restrict to shuffles of words on the same
alphabet, done by $2$-deterministic $3$-automata.  We prove that if
$x$ and~$y$ are normal words, $x$ and $y$ are finite-state independent
exactly when any shuffle of them is also normal.  The interleaving of the
symbols from $x$ and $y$ must be driven by a deterministic and oblivious
automaton reading $x$ and~$y$.  Here oblivious means that the choice of
inserting in the shuffled word~$z$ a symbol either from~$x$ or from~$y$ is
only made upon the current state of the automaton and not upon the current
symbols read from $x$ and~$y$.

\begin{definition}
  A \emph{shuffler} is a $2$-deterministic $3$-automaton such that each
  of its transitions has either the type $p \trans{a,\emptyword|a} q$
  (type~I) or the type $p \trans{\emptyword,a|a} q$ (type~II).
\end{definition}

Notice that the determinism of a shuffler~$\mathcal{S}$ implies that
for each of its states~$p$, all the transitions leaving~$p$ have the same type,
either type~I or type~II.  A transition of type~I copies a symbol from the
first input~$x$ to the  output and a transition of type~II copies a
symbol from the second input~$y$ to the  output.  It follows then that
the third word~$z = \mathcal{S}(x,y)$ is obtained by shuffling $x$ and~$y$.
This justifies the terminology.

\begin{figure}[htbp]
    \begin{center}
    \begin{tikzpicture}[->,>=stealth',initial text=,semithick,auto,inner sep=3pt]
      \tikzstyle{every state}=[minimum size=0.4]
      \node[state,initial above] (q0) at (0,0)  {$q_0$};
      \node[state]  (q1) at (3,0) {$q_1$};
      \path (q0) edge[out=210,in=150,loop] node {$0,\emptyword|0$} ();
      \path (q0) edge[bend left=20] node {$1,\emptyword|1$} (q1);
      \path (q1) edge[bend left=20] node {$\emptyword,1|1$} (q0);
      \path (q1) edge[out=30,in=-30,loop] node {$\emptyword,0|0$} ();
    \end{tikzpicture}
  \end{center}
\vspace*{-0.8cm}
  \caption{A shuffler}
  \label{fig:shuffler}
\end{figure}
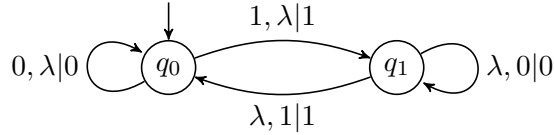

Consider infinite words $x = \overline{0011010001}\cdots$ and
$y = \underline{01000110001}\cdots $ and let $\mathcal{S}$ be the shuffler
pictured in Figure~\ref{fig:shuffler}.  Then, the infinite word
$z = \mathcal{S}(x,y)$ has the form
\begin{displaymath}
  z = \overline{001}\underline{01}\overline{1}
      \underline{0001}\overline{01}\underline{1}
      \overline{0001}\underline{0001}\cdots
\end{displaymath}
where the underlines and the overlines have been added to mark the origin
of each symbol.

If two normal words $x$ and~$y$ are on different  alphabets then,
in general, their shuffling $\mathcal{S}(x,y)$ is not normal.  
For instance, if $x$ and $y$ are words on different alphabets
their join is not normal.   Thus, we assume now a unique alphabet. 

Exchanging the input and output tapes of a shuffler~$\mathcal{S}$ gives
a $1$-deterministic $3$-automaton that we call the \emph{splitter}
corresponding to~$\mathcal{S}$.  This is due to the very special form of
the transitions of shufflers.  If the output $z = \mathcal{S}(x,y)$ of the
shuffler~$\mathcal{S}$ on inputs $x$ and~$y$ is fed to the corresponding
splitter, the two outputs are $x$ and~$y$.  The fact that the corresponding
splitter is $1$-deterministic yields the following lemma which really
requires that the alphabets on the two tapes are equal.

\begin{lemma} \label{lem:shuffler} 
  Let $\mathcal{S}$ be a shuffler and $q$ one of its states.  For each
  finite word~$w$, there is exactly one run of length~$|w|$ starting at~$q$
  and outputting~$w$.
\end{lemma}

\section{Proof of the  Characterization Theorem} \label{sec:proof}

\subsection{From independence to fairness  and back}

We first prove that finite-state independence implies fairness of the run in each
$2$-deterministic $2$-automata.

\begin{proof}[Proof of Theorem~\ref{thm:charac}, (\ref{sta:independence}) implies (\ref{sta:fairness})]
  For simplicity we assume that $A$ is the binary alphabet $\{ 0 ,1\}$ but
  the proof can easily be extended to the general case. We suppose by
  contradiction that there is a $2$-deterministic
  $2$-automaton~$\mathcal{A}$ such that the run on $x$ and~$y$
  in~$\mathcal{A}$ is not fair for edges and we claim that $x$ and $y$ are
  not finite-state  independent.

  By definition, each transition of~$\mathcal{A}$ is of the form
  $p \trans{a,\emptyword} q$ or $q \trans{\emptyword,b} q$ for some symbols
  $a$ and~$b$. For the rest of the proof, transitions of the form
  $p \trans{a,\emptyword} q$ are called of type~I and transitions of the
  form $p \trans{\emptyword,b} q$ are called of type~II.  Since the
  automaton is deterministic, all the transitions starting at each
  state~$q$ have the same type.  A state $q$ is said to be of type~I
  (respectively II) if all transitions starting at~$q$ have type~I
  (respectively II).

  We suppose that there is exists a state~$p$ and two transitions $\sigma$
  and~$\sigma'$ starting from~$p$ such that
  \begin{displaymath}
    \lim_{n\to\infty}{\frac{\occ{\gamma[1..n]}{\sigma}}{n}} \neq
    \lim_{n\to\infty}{\frac{\occ{\gamma[1..n]}{\sigma'}}{n}}.
  \end{displaymath}
  meaning that either at one of the two limits does not exist or that they
  both exist but they are not equal.  By symmetry it can be assumed that
  all transitions starting at~$p$, including $\sigma$ and~$\sigma'$, are of
  type~I.

  We show that $x$ can be compressed given~$y$.  There is a lack of
  symmetry between $x$ and~$y$ because transitions $\sigma$ and~$\sigma'$
  are of type~I.  By replacing $(k_n)_{n\ge0}$ by one of its subsequences,
  it can be assumed that, for each transition~$\tau$,
  $\lim_{n\to\infty}{\occ{\gamma[1..k_n]}{\tau}/k_n}$ exists and that
  $\lim_{n\to\infty}{\occ{\gamma[1..k_n]}{\sigma}/k_n} \neq
  \lim_{n\to\infty}{\occ{\gamma[1..k_n]}{\sigma'}/k_n}$.  Since the
  frequency of each state is equal to the sum of the frequencies of the
  transitions which start at it, the limit
  $\lim_{n\to\infty}{\occ{\gamma[1..k_n]}{q}/k_n}$ exists for each
  state~$q$.  Denote this limit by~$\pi(q)$.
 
  For each transition~$\tau$ starting at a state~$q$, let $\pi(\tau)$
  be defined as follows.
  \begin{displaymath}
    \pi(\tau) = 
    \begin{cases}\displaystyle
      \lim_{n\to\infty}\frac{\occ{\gamma[1..k_n]}{\tau}}{\occ{\gamma[1..k_n]}{q}}
      & \text{if $\lim_{n\to\infty}{\occ{\gamma[1..k_n]}{q}/k_n} \neq 0$} \\
     \displaystyle \frac{1}{2} & \text{otherwise}
    \end{cases}
  \end{displaymath}
  Since
  $\lim_{n\to\infty}{\occ{\gamma[1..k_n]}{\sigma}/k_n} \neq
  \lim_{n\to\infty}{\occ{\gamma[1..k_n]}{\sigma'}/k_n}$,
  $\pi(\sigma) \neq \pi(\sigma')$.  Furthermore, the following equality
  holds for each state~$q$.
  \begin{displaymath}
    \sum_{\text{$\tau$ starts at~$q$}}\pi(\tau) = 1.
  \end{displaymath}
  Since $x$ is normal it suffices to show that $\rho(x/y)<1$.  Let $\ell$
  be a block length to be fix later.  Let $\gamma$ be a finite run of
  length~$\ell$, so $\gamma$ is a sequence
  $\tau_1\tau_2 \cdots \tau_{\ell}$ of $\ell$ consecutive transitions.
  Let $\pi(\gamma)$ be defined as follows.
  \begin{displaymath}
    \pi(\gamma) = 
    \begin{cases}\displaystyle
      \prod_{\substack{\text{$\tau_i$ of type I}\\1\leq i\leq \ell}}{\pi(\tau_i)}
        & \text{if $\gamma$ has transition of type I} \\
      1 & \text{otherwise}
    \end{cases}
  \end{displaymath}
  Let $q$ be a state and $\bar{v}$ be a word of length~$\ell$.  Let
  $\Gamma_{q,\bar{v}}$ be the set of runs of length~$\ell$, starting at~$q$
  and reading a prefix of~$\bar{v}$ on the second tape,
  \begin{displaymath}
    \Gamma_{q,\bar{v}} = \{ \gamma : \gamma = q \trans{u,v} q', 
                                   v \sqsubset \bar{v}, |u|+|v| = \ell \}.
  \end{displaymath}
  Notice  that the sets $\Gamma_{q,\bar{v}}$ are not 
  always pairwise disjoint.  The word~$v$ read by the run~$\gamma$ on
  the second tape can be the prefix of several words $\bar{v}$.  If $v$ is
  the prefix of both $\bar{v}$ and~$\bar{v}'$, then the run~$\gamma$
  belongs to both $\Gamma_{q,\bar{v}}$ and~$\Gamma_{q,\bar{v}'}$.

  We claim that for each state~$q$ and each word~$\bar{v}$,
  \begin{displaymath}
    \sum_{\gamma \in \Gamma_{q,\bar{v}}}{\pi(\gamma)} = 1.
  \end{displaymath}
  We prove it by induction on the length of the run, that we call~$\ell$.
  If $\ell = 0$, the only run $\gamma \in~\Gamma_{q,\bar{v}}$ is the empty
  run so $\pi(\gamma)=1$.  Suppose now that $\ell \ge 1$.  We distinguish
  two cases.  First case: the transitions starting at $q$ are of type I.
  Suppose first that the transitions starting at~$q$ are the two
  transitions $\tau_0 = q \trans{0,\emptyword} q_0$ and
  $\tau_1 = q \trans{1,\emptyword} q_1$.  And suppose that
  $\bar{v} = \bar{v}'a$ where $\bar{v}' = \bar{v}[1..\ell-1]$ and $a$ is
  the last symbol of~$\bar{v}$.  The set $\Gamma_{q,\bar{v}}$ is then equal
  to the disjoint union
  $\Gamma_{q,\bar{v}} = \tau_0\Gamma_{q_0,\bar{v}'} \cup
  \tau_1\Gamma_{q_1,\bar{v}'}$.  The result follows from the inductive
  hypothesis since
  $\pi(\Gamma_{q,\bar{v}}) = \pi(\tau_0)\pi(\Gamma_{q_0,\bar{v}'}) +
  \pi(\tau_1)\pi(\Gamma_{q_1,\bar{v}'}) = \pi(\tau_0) + \pi(\tau_1) = 1$.
  Second case: the transitions starting at~$q$ have type~II.  Suppose that
  $\bar{v} = a\bar{v}'$ where $a$ is the first symbol of~$\bar{v}$ and
  $\bar{v}' = \bar{v}[2..\ell]$.  The transition
  $\tau = q \trans{\emptyword,a} q'$ is the first transition of each run
  in~$\Gamma_{q,\bar{v}}$ and
  $\Gamma_{q,\bar{v}} = \tau\Gamma_{q',\bar{v}'}$.  The result follows from
  the inductive hypothesis since
  $\pi(\Gamma_{q,\bar{v}}) = \pi(\Gamma_{q',\bar{v}'}) = 1$.  Since
  $\sum_{\gamma \in \Gamma_{q,\bar{v}}}{\pi(\gamma)} = 1$, there exists,
  for each state~$q$ and each word~$\bar{v}$, a prefix-free set
  $P_{q,\bar{v}} = \{w_{\gamma,\bar{v}} : \gamma \in \Gamma_{q,\bar{v}}\}$
  such that $|w_{\gamma,\bar{v}}| \le \lceil -\log \pi(\gamma) \rceil$
  holds for each run $\gamma \in \Gamma_{q,\bar{v}}$.  These words can be
  used to define a compressor~$\mathcal{C}$ which runs as follows on two
  inputs.  It simulates~$\mathcal{A}$ and it has $\ell$ symbols of look
  ahead on the second tape.  For each run~$\gamma$ of length~$\ell$, the
  compressor outputs $w_{\gamma,\bar{v}}$ on the third tape.  The choice
  of~$w_{\gamma,\bar{v}}$ depends on the look ahead~$\bar{v}$.

    We finally show that $\rho_{\mathcal{C}}(x/y) < 1$.
  The run $\gamma$ of~$\mathcal{A}$ on $x$ and~$y$ can be factorized
  as $\gamma = \gamma_1 \gamma_2 \gamma_3 \cdots$ where each run~$\gamma_i$
  has length~$\ell$.  The output of the compressor~$\mathcal{C}$ is then
  $w_{\gamma_1,\bar{v}_1}w_{\gamma_2,\bar{v}_2}w_{\gamma_3,\bar{v}_3} \cdots$
  where the words $\bar{v}_1, \bar{v}_2, \bar{v}_3, \ldots$ are the
  corresponding look ahead of $\ell$ symbols.  Let $\varepsilon,\delta>0$ be
  two positive real numbers.  Let $n$ be an integer large enough  such that
  $\occ{\gamma[1..k_n]}{\tau} \le (1+\delta) \pi(q)\pi(\tau)k_n$ for each
  transition~$\tau$ starting at~$q$. Then,
  \begin{align*}
    |w_{\gamma_1,\bar{v}_1}\cdots w_{\gamma_n,\bar{v}_n}| 
    & \le \sum_{i = 1}^n \lceil -\log \pi(\gamma_i) \rceil \\
    & \le n + \sum_{i = 1}^n -\log \pi(\gamma_i)  \\
    & \le n + \sum_{\text{$\tau$ of type I}}\occ{\gamma[1..\ell n]}{\tau}
                  \log \frac{1}{\pi(\tau)}  \\
     & \le \ell n \left[\frac{1}{\ell}+ (1+\delta)
         \sum_{\text{$q$ of type I}}\pi(q) 
         \sum_{\text{$\tau$ starts at~$q$}}\pi(\tau)
         \log \frac{1}{\pi(\tau)} \right] 
  \end{align*}
  Then, for each state~$q$, 
  \begin{displaymath}
  \sum_{\text{$\tau$ starts at~$q$}}\pi(\tau)\log \frac{1}{\pi(\tau)}
   \le 1
  \end{displaymath}
  and the relation is strict for $q = p$.  Since $\pi(p) > 0$, for
  $\varepsilon$ small enough, $\delta$ and $\ell$ can be chosen such that
  \begin{displaymath}
    \frac{1}{\ell}+ (1+\delta)\sum_{\text{$q$ of type I}}\pi(q)
    \sum_{\text{$\tau$ starts at~$q$}}\pi(\tau)\log \frac{1}{\pi(\tau)}
    \le (1-\varepsilon) \sum_{\text{$q$ of type I}}\pi(q).
  \end{displaymath}
  We obtain 
  \begin{displaymath}
    |w_{\gamma_1,\bar{v}_1}\cdots w_{\gamma_{k_n},\bar{v}_{k_n}}| \le 
    (1-\varepsilon)\ell k_n\sum_{\text{$q$ of type I}}\pi(q).
  \end{displaymath}
  Since $\sum_{\text{$q$ of type I}}\pi(q)$ is the limit of the ratio
  between the number of symbols read from~$x$ and the length of the run, we
  conclude $\rho_{\mathcal C}(x/y)<1$.
\end{proof}

We now prove  that fairness of the run in each $2$-deterministic
$2$-automata implies finite-state  independence.

\begin{proof}[Proof of Theorem~\ref{thm:charac}, (\ref{sta:fairness}) implies (\ref{sta:independence})]
  Let $x$ and~$y$ be two normal words such that
  statement~(\ref{sta:fairness}) of Theorem~\ref{thm:charac} holds.  We
  show that $x$ and~$y$ are finite-state independent.  It is sufficient to
  show that $x$ cannot be compressed with the help of~$y$, since the other
  incompressibility result is obtained by exchanging the roles of $x$
  and~$y$.

  Let $\mathcal{C}$ be a $2$-deterministic $3$-automaton such that for
  each~$y$, the function $x \mapsto \mathcal{C}(x,y)$ is one-to-one. By
  Lemma~\ref{lem:fairedgesscc}, the automaton~$\mathcal{C}$ can be assumed
  to be strongly connected.  Let $q_0$ be the initial state
  of~$\mathcal{C}$.  Let $\gamma$ be the run of~$\mathcal{C}$ on $x$
  and~$y$ and let $z$ be the output of~$\mathcal{C}$ along~$\gamma$, that
  is, $z = \mathcal{C}(x,y)$.  Let $\varepsilon > 0$ be a positive real
  number.  We claim that the compression ratio $\rho_{\mathcal{C}}(x/y)$
  satisfies $\rho_{\mathcal{C}}(x/y) > 1 - \varepsilon$. Since this holds
  for each $\varepsilon > 0$, this shows that
  $\rho_{\mathcal{C}}(x/y) \ge 1$.
  
  Let $k$ be a positive  integer to be fixed later.   Since $y$ is normal,
  there exists a constant~$K > 0$ such that if $u \sqsubset x$, $v \sqsubset
  y$ and $w \sqsubset z$ ($u$, $v$ and~$w$ are prefixes of $x$, $y$ and~$z$
  respectively) such~that 
  \begin{displaymath}
    q_0 \trans{u,v|w} q
  \end{displaymath}
  then $|v| \le K|u|$, see \cite[Lemma~5.3]{BCH2016}.
  The run~$\gamma$ is decomposed
  \begin{displaymath}
    q_0 \trans{u_1,v_1|w_1} 
    q_1 \trans{u_2,v_2|w_2} 
    q_2 \trans{u_3,v_3|w_3} \cdots
  \end{displaymath}
  where $|u_i| = k$ for each integer $i \ge 1$. Note that the lengths of each
  word~$v_i$ and each word~$w_i$ are arbitrary.  Our aim is to prove that for
  $N$ large enough $|w_1 \cdots w_N| \ge (1-4\varepsilon) |u_1 \cdots u_N|$.

  Let $\ell$ be the integer $\lceil kK/\varepsilon \rceil$.  By definition
  of~$\ell$, the cardinality of the set $\{ i \le N : |v_i| > \ell \}$ is
  less than $\varepsilon N$.  Otherwise we would have
  $|v_1 \cdots v_N| > K|u_1 \cdots u_N|$ which contradicts the definition of
  the constant~$K$.
  The indices $i$ such that $|v_i| > \ell$ are ignored in the sequel.

  Let $v'_i$ be the prefix of length~$\ell$ of the infinite
  word~$v_iv_{i+1}v_{i+2} \cdots$.  Unless $|v_i| > \ell$, $v_i$ is a
  prefix of~$v'_i$.  Let $v' \in B^{\ell}$ be a fixed word of
  length~$\ell$.  We claim that the cardinality of the set
  \begin{displaymath}
    X_{v'} = \{ u \in A^k : \exists p,q \;\; p \trans{u,v|w} q,  v \sqsubset v'
     \text{ and } |w|<(1-\varepsilon)k\}
  \end{displaymath}
  is bounded by $(\ell+1)|Q|^2|A|^{k(1-\varepsilon)}$.  For each choice of
  $p$, $q$, $v$ and $w$, there is at most one possible~$u$.  Otherwise, the
  function $x \mapsto \mathcal{C}(x,y)$ would not be one-to-one.  The terms
  $(\ell+1)$, $|Q|^2$ and $|A|^{k(1-\varepsilon)}$ account respectively for
  the number of choices for $v$, $p$ and~$q$, and $w$.  Note that the
  number of choices of~$v$ is $\ell+1$ because $v$ is a prefix of the fixed
  word~$v'$ of length~$\ell$.  The integer~$k$ is chosen such that
  $|A|^k-(\ell+1)|Q|^2|A|^{k(1-\varepsilon)}$ is greater than
  $(1-\varepsilon)|A|^k$.  This is possible because $|Q|$ is constant and
  $\ell$ grows linearly with~$k$.

  By fairness and by Lemma~\ref{lemma:2}, it follows that for $N$ great enough
  and for any words $u \in A^k$ and $v' \in A^{\ell}$
  \begin{displaymath}
    \#\{i : u_i = u \text{ and } v'_i = v' \} \ge (1-\varepsilon)N/|A|^{k+\ell}.
  \end{displaymath}
  Summing up for all $u \notin X_{v'}$ and all $v' \in A^{\ell}$ gives that
  for $N$ great enough 
  \begin{displaymath}
    \#\{i : u_i \notin X_{v'_i} \} \ge  (1-\varepsilon)^2N,
  \end{displaymath}
  and subtracting the number of~$i$ such that $|v_i| \ge \ell$ gives
   \begin{displaymath}
    \#\{i : u_i \notin X_{v'_i} \text{ and } v_i \sqsubset v'_i\} \ge
    [(1-\varepsilon)^2-\varepsilon]N \ge  (1-3\varepsilon)N.
  \end{displaymath}
  For each $i$ in the previous set, $w_i \ge (1-\varepsilon)k$.  Therefore
  the length of the output $w_1 \cdots w_N$ is at least 
  $(1-3\varepsilon)(1-\varepsilon)kN \ge (1-4\varepsilon)kN$.
  This completes the proof since the length of the input $u_1 \cdots u_N$
  is $kN$.
\end{proof}

\subsection{From independence/fairness to selecting and back}

\begin{proof}[Proof of Theorem~\ref{thm:charac}, (\ref{sta:independence}) implies (\ref{sta:selector})] 
  We need to prove that selection from a normal word~$x$ with a
  finite-state independent normal oracle~$y$ preserves normality.  Mutatis mutandis this
  proof is the same as that given
  in~\cite[Theorem~7.1]{BecherCartonHeiber15}, but now one should consider
  $2$-deterministic $3$-automata, and the normal word $y$ as a consultative
  oracle.
\end{proof}

\begin{proof}[Proof of Theorem~\ref{thm:charac}, (\ref{sta:selector}) implies 
(\ref{sta:fairness})]
  Suppose  that fairness does not hold.  By Proposition~\ref{pro:fairness},
  there is a $2$-deterministic automaton~$\mathcal{A}$ with the following
  property.  Let $\gamma$ be the run of~$\mathcal{A}$ on $x$ and~$y$.
  There are in~$\mathcal{A}$ and two transitions $\tau$ and~$\tau'$
  starting from the same state~$p$ and an increasing sequence
  $(k_n)_{n\ge0}$ of integers such that
  \begin{displaymath}
    \lim_{n\to\infty}{\frac{\occ{\gamma[1..k_n]}{\tau}}{n}}      \neq 
    \lim_{n\to\infty}{\frac{\occ{\gamma[1..k_n]}{\tau'}}{n}}.
  \end{displaymath}
  Since $\mathcal{A}$ is $2$-deterministic, all transitions starting at~$q$
  read symbols from the same tape.  The automaton~$\mathcal{A}$ can be
  turned into a selector~$\mathcal{S}$ as follows.  Transitions starting
  at~$p$ select the digit they read but all other transitions do not select
  the digit they read.  The previous inequality shows that the output of
  the selector~$\mathcal{S}$ is not even simply normal.  This is a
  contradiction with the hypothesis.
\end{proof}

We end this section with the following result that shows that the 
finite-state independence of two normal words implies the finite-state  independence of one and a 
word that results from selection of the other.

\begin{proposition}
  Let $x$ and $y$ be normal and finite-state independent words.  If $y'$ is obtained by
  oblivious selection from~$y$, then $x$ and~$y'$ are still finite-state independent.
\end{proposition}

\begin{proof}
  We show that if $x$ and~$y'$ are not finite-state independent, then $x$ and~$y$ are
  also not finite-state independent.  We suppose that $x$ and~$y'$ are not finite-state independent.
  This means either that $x$ can be compressed with the help of~$y'$ or
  that $y'$ can be compressed with the help of~$x$.  Suppose first that $x$
  can be compressed by a compressor~$\mathcal{C}$ with the help of~$y'$.
  Combining this compressor with the selector~$\mathcal{S}$ which selects
  $y'$ from~$y$ yields a compressor~$\mathcal{C}'$ which compresses $x$
  with the help~$y$.  Indeed, this compressor~$\mathcal{C}'$ skips symbols
  from~$y$ which are not selected by~$\mathcal{S}$ and
  simulates~$\mathcal{C}$ on those symbols which are selected
  by~$\mathcal{S}$.

  Suppose second that $y'$ can be compressed by a compressor~$\mathcal{C}$
  with the help of~$x$.  We claim that $y$ can also be compressed with the
  help of~$x$.  The selector~$\mathcal{S}$ which selects $y'$ from~$y$ is
  used as a splitter to split $y$ into $y'$ made of the selected symbols
  and $y''$ made of the non selected symbols.  Then, the
  compressor~$\mathcal{C}$ is used to compress $y'$ with the help of~$x$
  into a word~$z$.  Finally, words $z$ and~$y''$ are merged into a
  word~$z'$ by blocks of the same length~$m$.  Each block of length~$m$
  contains either $m$ symbols from~$z$ or $m$ symbols from~$y''$ plus an
  extra symbol indicating whether the block contains symbols from~$z$ or
  symbols from~$y''$.  The combination of all these automata yields a
  compressor which compresses $y$ with the help of~$x$.
\end{proof}

\subsection{From independence/fairness to shuffling and back}

%
%

\begin{proof}[Proof of Theorem~\ref{thm:charac}, (2) implies (4)]
  Suppose $x$ and $y$ are normal.
  Let $\gamma$ be the run of the shuffler~$\mathcal{S}$ with inputs $x$ and~$y$
  and let $\ell$ be a given length.  For each state~$q$ of~$\mathcal{S}$ and
  each word~$w$ of length~$\ell$, there exists by Lemma~\ref{lem:shuffler} a
  unique run~$\rho_{q,w}$ starting at state~$q$ and outputting~$w$.  

  For each word~$w$ of length~$\ell$, the number of occurrences of~$w$ in the
  prefix $z[1..n]$ of~$z$ is given by
  \begin{displaymath}
    \occ{z[1..n]}{w} = \sum_{q\in Q} \occ{\gamma[1..n]}{\rho_{q,w}}.
  \end{displaymath}
  By Lemma~\ref{lem:fairedgesscc}, the run~$\gamma$ reaches a recurrent
  strongly connected component.  Thus, it can be assumed without loss of
  generality that $\mathcal{S}$ is strongly connected.  By Lemmas
  \ref{lemma:1} and~\ref{lemma:2}, for any two finite runs $\rho$
  and~$\rho'$ of the same length and starting from the
  same state, the following equality holds.
  \begin{displaymath}
    \lim_{n\to\infty}{\frac{\occ{\gamma[1..n]}{\rho}}{n}}      = 
    \lim_{n\to\infty}{\frac{\occ{\gamma[1..n]}{\rho'}}{n}}.
  \end{displaymath}
  The result follows directly from this equality.
\end{proof}

\begin{proof}[Proof of Theorem~\ref{thm:charac}, (4) implies (1)]
Suppose that $x$ and~$y$ are not finite-state independent and
$x$ is compressible with the help of~$y$.  Let $\mathcal{A}$ be the
compressor such that $\rho_{\mathcal{A}}(x/y) < \rho(x)$.  Consider the
shuffler $\mathcal{S}$ that mimics $\mathcal{A}$ and copies each digit
of~$x$ (respectively of $y$) as soon as it is read by ${\mathcal A}$.  
We claim that $\mathcal{S}(x,y)$ is compressible, hence not normal. 
For compressing $\mathcal{S}(x,y)$, first define a splitter  $\mathcal{S}'$ 
exchanging the inputs and outputs in the transition of $\mathcal{S}$.
Thus, $\mathcal{S}'(\mathcal{S}(x,y))=(x,y)$.
By composing  $\mathcal{S}' $ with ${\mathcal A}$
we can   compress $x$  using~$y$ and obtain a compressed word  $x'$.
Let $m$ be the block size used in this compression.
Finally, words $y$ and~$x'$ are merged into a  word~$z$ 
interleaving  a block of $m$ symbols from $x$ with a block of $m$ symbols from $y$.
Since the hypothesis ensures $x$ is compressible, so is word $z$.
From this word $z$ we can  recover $(x',y)$, 
from which we can recover  $(x,y)$ and  then obtain $S(x,y)$, as required.
\end{proof}

\section{An algorithm for a pair of independent normal words}
\label{sec:algorithm}

To prove Theorem~\ref{thm:algorithm} we give an explicit algorithm based on
the characterization of finite-state independent normal words in terms of
shufflers (Theorem~\ref{thm:charac} statement~(4)).  The algorithm we
present here is an adaptation of Turing's algorithm for computing an
absolutely normal number~\cite{turing,BFP2007}.  But instead of computing
the expansion of a number that is normal in every integer base here we
compute a pair of normal infinite words such that every shuffling of them
produced by a finite-state shuffler is normal.  We start with auxiliary
definitions and some properties.  We write $\log$ for the logarithm in base
$e$ and $\log_b$ for any other base $b$.

\begin{definition} \label{def:sets}
1. \ For a shuffler ${\mathcal S}$, a real number $\varepsilon > 0$, a
finite word $\gamma \in A^*$ and a positive integer $n$, 
we define the set 
\begin{align*}
E_{{\mathcal S}}(\varepsilon, \gamma, n)&=\left\{(x,y) \in A^\omega \times A^\omega
:\left| \occ{{\mathcal S}(x,y)[1..n]}{\gamma} - {n}/{|A|^{|\gamma|}} \right| < \varepsilon n\right\}.
\end{align*}
2.\ Assume an enumeration of shufflers ${\mathcal S}_1, {\mathcal S}_2, \ldots$ and 
define the set
\begin{align*}
F(\varepsilon, t, \ell, n) &= 
\bigcap_{i=1}^t \bigcap_{r=1}^\ell \bigcap_{\gamma \in A^r} E_{{\mathcal S}_i}(\varepsilon, \gamma, n).
\end{align*}
3.\ For  each positive integer $n$, let
    $\ell_n = ({\log_{|A|} n})/{3}$, 
    $t_n = n$ and 
    $\varepsilon_n = 2 \;\;\sqrt{(\log n \log_{|A|} n)/n}$. 
\[
F_n = F(\varepsilon_n, t_n, \ell_n, n). 
\]
\end{definition}

\begin{lemma}[Lemma~8 in~\cite{turing}, adapted from Theorem 148 in~\cite{HardyWright2008}] \label{lemma:hardy}
Let $r$ and $n$ be positive integers. 
For every real $\varepsilon$ such that ${6}/{\lfloor n/r \rfloor} \leq \varepsilon \leq {1}/{|A|^ r}$ 
and for every $\gamma \in A^r$, 
if  $N(\gamma, i, n) = |\{w \in A^n : \occ{w}{\gamma} = i \}|$ then
\[
\sum_{0\leq i\leq  n/|A|^r - \varepsilon n} N(\gamma,i,n)
+
\sum_{ n/|A|^r + \varepsilon n\leq i\leq n} N(\gamma,i,n)
<  \ 2 |A|^{n+2r-2} r e^{-{|A|^r \varepsilon^2 n}/{6r}}.
\]
\end{lemma}

For a word $u \in A^*$ we denote by  $[u]$ the set of infinite words 
that start with $u$, and we call it the  cylinder determined by $u$, 
\[
[u] = \{x \in A^\omega : x[1..|u|] = u\}.
\]
For the Cartesian product of two cylinders  $[u] \times [v]$ 
we write   $([u], [v])$, and we call  the pair of cylinders
 determined by  $(u,v)$.

\begin{proposition} \label{prop:boundE} 
For every shuffler ${\mathcal S}$, every  $n, r,\varepsilon$
such that ${6}/{\lfloor n/r \rfloor} \leq \varepsilon \leq {1}/{|A|^ r}$
and every $\gamma\in A^r$,
\[
\mu (E_{\mathcal S}(\varepsilon, \gamma, n)) > 
1 - 2 |A|^{2r-2} r e^{-{|A|^r \varepsilon^2 n}/{6r}}.
\]
\end{proposition}
\begin{proof}
Consider the set 
\begin{align*}
P(\varepsilon, \gamma, n)&=\left\{w\in A^n: 
\left|\occ{w}{\gamma} - {n}/{|A|^{|\gamma|}}\right| < \varepsilon n\right\}.
\\\omit\hspace*{-.9cm}\text{Then, }
\\
E_{\mathcal S}(\varepsilon, \gamma, n) 
& =\bigcup_{w \in P(\varepsilon, \gamma, n)}  
\{([u],[v])  : |u| + |v| = n \text{ and }\forall x\in [u]\forall y\in [v],\  {\mathcal S}(x,y) \in [w]\}
\\
&= \bigcup_{w \in P(\varepsilon, \gamma, n)}  {\mathcal S}^{-1}([w] ).
\\\omit\hspace*{-.9cm}\text{Thus, }
\\
\mu(E_{{\mathcal S}}(\varepsilon, \gamma, n))& = 
\sum_{w \in P(\varepsilon, \gamma, n)} \mu({\mathcal S}^{-1}([w])) = 
|P(\varepsilon, \gamma, n)| \  |A|^{-n}.
\end{align*}
Finally, Lemma~\ref{lemma:hardy} gives the needed upper bound for 
$|\overline{P}(\varepsilon, \gamma, n)|$.
\end{proof}

For any set $B\subseteq A^\omega\times A ^\omega$ we write $\overline{B}$ to denote its complement,
$(A^\omega\times A ^\omega)\setminus B$.

\begin{proposition} \label{prop:boundA}
For any $\varepsilon$, $t$, $\ell$ and $n$, such that ${6}/{\lfloor n/\ell \rfloor}
 \leq \varepsilon \leq {1}/{|A|^\ell}$,
\[
\mu(F(\varepsilon, t, \ell, n)) > 1 - 2 t |A|^{3\ell-1} e^{-{\varepsilon^2 n}/(3\ell)}.
\]
\end{proposition}
\begin{proof}By Definition \ref{def:sets},
\[
\mu (\overline{F}(\varepsilon, t, \ell, n)) \leq 
\sum_{i = 1}^t \sum_{r = 1}^{\ell} \sum_{\gamma \in A^r} 
\mu(\overline{E_{{\mathcal S}_i}}(\varepsilon, \gamma, n)).
\]
The number of terms of this triple sum is bounded by
\[
\sum_{i = 1}^t \sum_{r = 1}^{\ell} \sum_{\gamma \in A^r} 1 =  
\sum_{i = 1}^t \sum_{r = 1}^{\ell} |A|^r < 
\sum_{i=1}^t\frac{|A|^{\ell+1} - 1}{|A|-1} < 
\sum_{i=1}^t |A|^{\ell+1} =
t |A|^{\ell+1}.
\]
From the lower bound given in  Proposition~\ref{prop:boundE}
we obtain that for every shuffler ${\mathcal S}$ and  for every  
word $\gamma\in A^{\leq \ell}$,
\[
\mu(\overline{E_{\mathcal S}}(\varepsilon, \gamma, n)) < 
2 |A|^{2\ell-2} \ell e^{-{\varepsilon^2 n}/(3\ell)}.
\]
Therefore,
\[
\mu(\overline{F}(\varepsilon, t, \ell, n)) < 
2 t |A|^{3\ell-1} e^{-{\varepsilon^2 n}/(3\ell)}.\qedhere
\]
\end{proof}

Recall the values given in~Definition~\ref{def:sets} 
    $\ell_n = ({\log_{|A|} n})/{3}$, 
    $t_n = n$,
    $\varepsilon_n = 2 \sqrt{(\log n \log_{|A|} n)/n}$ 
and $F_n = F(\varepsilon_n, t_n, \ell_n, n)$.

\begin{proposition}\label{prop:boundquad}
Let $n_{\text{\em start}}=\min\{ n :  \varepsilon_n \geq 6 / \lfloor n/\ell_n \rfloor\}$.
Then for every $n\geq n_{\text{\em start}}$,  $\ell_n, t_n \geq 1$, 
\[
\mu(F_n) \geq 1 - {1}/{n^2}.
\]
\end{proposition}
\begin{proof}
To apply Proposition~\ref{prop:boundA} it is required that 
${6}/{\lfloor n/\ell_n \rfloor} \leq \varepsilon_n \leq {1}/{|A|^{\ell_n}}$.
Then, for every $n\geq n_{\text{\em start}}$ the required inequality holds. So,  
application of Proposition~\ref{prop:boundA} yields
\begin{align*}
\mu(\overline{F_n}) & \leq  2 t_n |A|^{3\ell_n-1} e^{-{\varepsilon^2 n}/(3\ell_n)}
\\
& \leq  t_n\ |A|^{3\ell_n} e^{-{\varepsilon^2 n}/(3\ell_n)} 
\\
& =  n |A|^{( \log_{|A|} n)}  e^{- 4n  (\log n )   (\log_{|A|} n)/ (n \log_{|A|} n)}  
\\
&=  n^2\       e^{-4 \log n}
\\
& = \frac{1}{n^2}.\qedhere
\end{align*}
\end{proof}

If  $n_{start}$ is as determined by Proposition~\ref{prop:boundquad},
then  $\bigcap_{n\geq n_{start}} F_n$  is not empty and
consists just of  pairs of finite-state  independent normal words.
We can actually show that the intersection of  a subsequence of $F_n$'s 
with $n$ growing at most exponentially, 
also consists just  of  pairs of finite-state  independent normal words.
The next definition  fixes $n_0$ as  $\log n_{\text{\em start}}$ and defines the 
sets~$G_n$ which are used in the proof of Theorem~\ref{thm:algorithm}.

\begin{definition}\label{def:G}
Let $n_0= \log_{|A|}\min\{ n :  \varepsilon_n \geq 6 / \lfloor n/\ell_n \rfloor\}$.
We define a sequence $(G_n)_{n \geq 0}$  of finite sets of 
pairs of cylinders in $A^\omega\times A^\omega$,
such that for every $n$, $G_{n+1} \subseteq G_n$ as 
\[
G_n = \bigcap_{j = 0}^{n} F_{|A|^{n_0 + j}}
\]
\end{definition}

\begin{lemma}\label{lemma:G}
The set $\bigcap_{n\geq 0} G_{n}$ 
consists exclusively of pairs of finite-state independent normal  words.
\end{lemma}

\begin{proof}
Fix $n_0$ as defined in Definition~\ref{def:G}.  
Suppose $(u,v)\in \bigcap_{n\geq 0} G_n$.
To show that $u$ and $v$ are finite-state independent we show that for any
shuffler ${\mathcal S}$,  ${\mathcal S}(u,v)$ is a normal sequence.
Fix a finite word  $w \in A^*$.
Pick $m_0$ such that if $i$ is the index of ${\mathcal S}$ in the enumeration of shufflers, 
$t_{m_0} \geq i$, $\ell_{m_0} \geq |w|$, $m_0 \geq n_0$ and  $\varepsilon_{m_0} < 1 / |A|^{|w|}$. 

Let's see that for any   $m$ greater than $m_0$ 
the following holds.
Let $k$ be such that  $|A|^k \leq m < |A|^{k+1}$. Then,
using that $(u,v)\in F_{|A|^{k+1}}$,
\begin{align*}
\frac{\occ{{\mathcal S}(u,v)[1..m]}{w}}{m} 
         &< \frac{\occ{{\mathcal S}(u,v)[1..|A|^{k+1}]}{w}}{m} \\
         & <  \frac{1}{m} |A|^{k+1} \left( \frac{1}{|A|^{|w|}} + \varepsilon_{m_0}\right) \\
        & \leq  \frac{|A|^{k+1}}{|A|^k} \frac{2}{|A|^{|w|}}\\
        & =  \frac{2|A|}{|A|^{|w|}}.
\end{align*}
This implies that 
\[
\limsup_{m\to\infty} \frac{\occ{{\mathcal S}_i(u,v)[1..m]}{w}}{m} < \frac{2|A|}{|A|^{|w|}}.
\]
We conclude that ${\mathcal S}(u,v)$ is normal 
applying Theorem~4.6 in~\cite{Bugeaud12} 
which establishes that 
a word $x$ is normal if, and only if, 
there exists a positive number $C$ such that  for every  finite word $w$,
\[
\limsup_{m \to \infty} \frac{\occ{x[1..m]}{w}}{m} \leq \frac{C}{|A|^{|w|}}.
\]
Hence,  taking $C$ equal to $2|A|$ we obtain that $S(u,v)$ is normal.
Now we  prove that  both, $u$ and $v$, are normal too.
Consider the selector ${\mathcal S}'$ defined as the 
 splitter     that reverses  ${\mathcal S}$  and then ignores the second output tape.
That is, if ${\mathcal S}(u,v)=z$ then ${\mathcal S}'(z) = u$. 
Since ${\mathcal S}(u,v)$ is normal, by Agafonov's theorem $u$ is normal.
A similar  argument proves  that $v$ is  also normal.
We proved that  every $(u,v)\in \bigcap_{n\geq 0} G_{n}$ is a pair of normal 
words satisfying statement~(4) of Theorem~\ref{thm:charac}.
Hence, $(u,v)$ is a pair of finite-state independent normal words.
\end{proof}

\setcounter{algocf}{\value{theorem}}

\begin{algorithm}

\DontPrintSemicolon
\SetKwFor{Repeat}{repeat}{}{forever}
\SetKwIF{If}{Elseif}{Else}{if}{then}{else if}{else}{\vspace*{-0.5cm}}

\SetKwBlock{Begin}{begin}{end}%
\SetKw{Print}{print}

\BlankLine

{\em Proof of Theorem~\ref{thm:algorithm} }\;

\BlankLine

\hspace*{-0.35cm}
\begin{tabular}{ll}
{\bf Algorithm}: & {Construction of a pair of normal finite-state  independent words}
\\
{\bf Input}: & No input
\\
{\bf Output}:& A sequence  $(I_n)_{n\geq 0} = ([u_n], [v_n])_{n\geq 0}$,
such that $u_n,v_n\in \{0,1\}^*$,
\\
& $|u_n|+|v_n|=n$ 
 and  $\bigcap_{i\geq 0} I_n $ contains a unique pair $(u,v)$ of
\\&
 finite-state independent normal words.
\end{tabular}

\BlankLine
\BlankLine

Let  ${\mathcal S}_1, {\mathcal S}_2, \ldots$ be a enumeration of shufflers.\;

For each $n\geq 1$, let 
$\ell_n = (\log n)/3$, 
$\varepsilon_n = 2 \sqrt{(\log n \log_2 n)/n}$ and
\begin{align*}
&F_n = \bigcap_{i=1}^n  \bigcap_{\gamma \in 2^{\leq \ell_n}}
E_{{\mathcal S}_i}(\varepsilon_n, \gamma, n), \text{ where }
\\
&E_{{\mathcal S}_i}(\varepsilon_n, \gamma, n)=\{
(x,y) \in \{0,1\}^\omega\times\{0,1\}^\omega
: \left| \  \occ{{\mathcal S}_i(x,y)[1..n]}{ \gamma} - n/2^{|\gamma|} \right| < n\varepsilon_n \}.
\end{align*}
Let  $n_0=\log_{2}\min\{ n :  \varepsilon_n \geq 6 / \lfloor n/\ell_n \rfloor\}$.
We write $\lambda$ for the empty word.

\BlankLine

\Begin{
$n \gets 0$\;
$I_0 \gets ([\lambda], [\lambda])$\;
$G_0 \gets ([\lambda], [\lambda])$\;

\Repeat{}{
    $([u_n],[v_n])\gets I_n$\;
    \If{$n$ is even} {
        $I_n^0 \gets ([u_n 0], [v_n])$\;
        $I_n^1 \gets ([u_n 1], [v_n])$\;
    }
    \Else{
        $I_n^0 \gets ([u_n], [v_n 0])$\;
        $I_n^1 \gets ([u_n], [v_n 1])$\;
    }
   \BlankLine

   $G_{n+1} \gets G_n \cap F_{2^{n_0 + n+1}}$;
   \BlankLine

    \If{$\mu(I_n^0 \cap G_{n+1}) > 2^{-2n+1}$} {
        $I_{n+1} \gets I_n^0$\;
    }
    \Else{
        $I_{n+1} \gets I_n^1$\;
    }
   \BlankLine
    \Print{$I_{n+1}$}\;

    $n \gets n+1$\;
}
}

\caption{Construction of a pair of normal finite-state independent words using shufflers}
 \label{alg:pair}
\end{algorithm}

\setcounter{theorem}{\value{algocf}}

\begin{proof}[Proof of Theorem~\ref{thm:algorithm}]
For clarity we present  the proof  for the alphabet $A=\{0,1\}$, hence $|A|=2$.
It is straightforward transfer the proof to any alphabet of an arbitrary size.
We  prove that Algorithm~\ref{alg:pair}
constructs of a pair of finite-state independent normal words.
From the algorithm is immediate that the sequence  
$(I_n)_{n \geq 0}$ is such that for every $n$, $I_{n+1} \subset I_n$,  and 
$\mu(I_{n+1}) = {\mu(I_n)}/{2}$.
We show that  for every $n$,
$\mu(I_n \cap G_n) > 0$.
We  prove by induction that for every $n$, 
\[
\mu(G_n \cap I_n) > 2^{-2n-1}.
\]
For the base case, $n = 0$, 
$\mu(G_0 \cap I_0) = 1 > 2^{-1}$.
For the inductive step, $n+1$, since 
\[
\mu(\overline{F_{2^{n_0 + n+1}}}) < \frac{1}{(2^{n_0+n+1})^2} = 2^{-2(n_0+n+1)} < 2^{-2(n+1)},
\]
we have
\begin{align*}
\mu(G_{n+1} \cap I_n) & =  \mu(G_n \cap I_n \cap F_{2^{n_0+n+1}}) \\
    & > 2^{-2n-1} - 2^{-2(n+1)}\\
    & = 2^{-2(n+1)}.
\end{align*}
Then, at least one of $G_{n+1} \cap I_n^0$ and $G_{n+1} \cap I_n^1$ 
must have measure greater than $2^{-2(n+1)-1}$, as required.
Since 
$(I_n)_{n\geq 0}$ is a nested sequence of intervals  of strictly decreasing but positive measure, 
and for every $n$, $\mu(G_n \cap I_n) > 0$,
we conclude that
\[
\bigcap_{n\geq 0} I_n = \bigcap_{n\geq 0} G_n \cap I_n
\]  
contains a unique pair $(u,v)$.
And by Lemma~\ref{lemma:G} all the elements in $\bigcap_{n\geq 0} G_n$
are pairs of  finite-state  independent normal words. This concludes the proof.
\end{proof}

\subsection{Computational complexity}

Algorithm~\ref{alg:pair} computes  a sequence  $(I_n)_{n\geq 0}$ 
 of pairs of cylinders in $\{0,1\}^\omega\times\{0,1\}^\omega$
 such that $\bigcap_{i\geq 0} I_n $ contains 
a unique pair $(u,v)$ of finite-state independent words.
We now establish its computational complexity.

\begin{proposition}
  Algorithm~\ref{alg:pair} has doubly exponential complexity: to output $n$
  symbols of the finite-state independent normal words $u$ and $v$ the algorithm
  performs a number of mathematical operations that is doubly exponential
  in~$n$.
\end{proposition}

\begin{proof}
  As in Turing's original construction, the complexity of each step of our
  algorithm is dominated by the computation of the set
  $F_{n_0 + 2^{2^n+1}}$, which is doubly exponential.  Notice that the
  measures of the inspected sets can be calculated in simply exponential
  time, and the rest of the computation takes constant time.

The construction works by taking a sequence of ``good sets'' $(G_n)_{n\geq 0}$ 
and a sequence $(I_n)_{n\geq 0}$ of pairs of cylinders in $\{0,1\}^\omega\times \{0,1\}^\omega$.
For the initial step, $n=0$, $\mu(G_0) = 1, \mu(I_0) = 1$, and $\mu(G_0 \cap I_0) = 1$.
For subsequent steps, we refine $G_n$ into $G_{n+1}$ and choose 
one suitable half of $I_n$ to be $I_{n+1}$.
We now find out the length $s_n$ of the shuffling that need to be inspected at step $n$ of the algorithm.
At step $n$,  $G_{n+1} = G_n \cap F_{s_n}$
and 
$\mu(G_{n+1}) \geq \mu(G_n) - \mu(\overline{F_{s_n}})$.
The algorithm chooses  the  half of $I_n$ whose  intersection with $G_{n+1}$
is  at least  $(\mu(G_n) - \mu(\overline{F_{s_n}}))/2$.
We need that for each $n$, this measure is  positive:
\begin{align*}
\big(  (  ( (\mu(G_0)-\mu(\overline{F_{s_0}}))/2-\mu(\overline{F_{s_1}}))/2- \mu(\overline{F_{s_2}}))/2 \ldots- \mu(\overline{F_{s_{n-1}}}) \big)/2 &> 0
\\
2^{-n} - 2^{-(n-1)}\mu(\overline{F_{s_0}}) - \ldots - 2^{-1} \mu(\overline{F_{s_{n-1}}})& > 0
\\
\hspace*{-8cm}\text{Multiplying by $2^n$}\hspace*{8cm}&
\\
1 - 2\mu(\overline{F_{s_0}}) - \ldots - 2^{n-1} \mu(\overline{F_{s_{n-1}}}) &> 0
\\
\sum_{n=1}^{\infty} 2^n \mu(\overline{F_{s_{n-1}}}) &< 1.
\end{align*}
Therefore, we require 
$\sum_{n=1}^{\infty} 2^n \mu(\overline{F_{s_{n-1}}}) < 1$ while
 Proposition~\ref{prop:boundquad} 
establishes that
\linebreak
 $\mu(\overline{F_{s_{n-1}}}) < 1/s_{n-1}^2$.
Thus,  we require    $s_{n-1}\geq 2^n$, which shows 
the needed exponential growth in the index of the sets $F_{s_n}$.
Notice that the  algorithm fixes  $s_{n}=2^{n+1}$
and  the computation of the set $F_{s_n}$ requires the inspection of 
$2^{s_n}$ words of length $s_{n}$.  
Then at step $n$ the algorithm performs a number of operations 
that is doubly exponential in $n$.
Finally notice that at step $n$ the algorithm outputs $n$ symbols
in the form of two words $u_n$, $v_n$, such that~$|u_n|+|v_n|=n$.
%
\end{proof}

\section{Open problems}\label{sec:conclusion}

As a conclusion, we would like to mention a few open problems.
\begin{enumerate}
\item The characterization of finite-state independence of normal words
  given in Theorem~\ref{thm:charac} uses asynchronous deterministic finite
  automata with no extra memory (counters, stack).  Determine if the same
  characterization holds for the non-deterministic version of the same
  finite automata.  We have pursued this line of investigation
  in~\cite{BecherCartonHeiber15,CartonHeiber15} for the characterization of
  normality in terms of incompressibility by finite-automata and
  essentially we found that, without extra memory, non-determinism, two-way
  does not add compressibility power.

\item Give a purely combinatorial characterization of finite-state
  independence of normal words.  We aim at a condition on the two sequences
  that is defined in combinatorial terms, without mentioning automata (in
  the same way that the definition of normality can be stated in terms of
  frequency of blocks).

\item There are efficient algorithms that compute absolutely normal numbers
  with nearly quadratic complexity as~\cite{BHS2013} or, as recently
  announced, in poly-logarithmic linear complexity~\cite{LM2016}.  It may
  be possible to adapt those algorithms to efficiently compute a pair of
  finite-state independent normal sequences.

\item Construct a normal word that is finite-state independent of some
  given normal word.  That is, given a word that has been proved to be
  normal, as Champernowne's word, we aim to construct another normal word
  that is finite-state independent of it.
\end{enumerate} 
\bigskip
\bigskip

\noindent
\textbf{Acknowledgements.}

The authors are members of the Laboratoire International Associ\'e INFINIS,
CONICET/Universidad de Buenos Aires–CNRS/Universit\'e Paris Diderot and
they are partially supported by the ECOS project PA17C04.  Carton is
partially funded by the DeLTA project (ANR-16-CE40-0007).

\bibliographystyle{plain}
\bibliography{sss}
\medskip

{\setstretch{0.9}
\begin{minipage}{\textwidth}
\noindent
Nicol\'as Alvarez\\
ICIC - Universidad Nacional del Sur, CONICET\\
Departamento de Ciencias en Ingenier\'ia de la Computaci\'on\\
naa@cs.uns.edu.ar
\\
\medskip\\
Ver\'onica Becher
\\
 Departamento de  Computaci\'on,   Facultad de Ciencias Exactas y Naturales
\\
Universidad de Buenos Aires \& ICC,  CONICET, Argentina.
\\
vbecher@dc.uba.ar
\\
\medskip\\
Olivier Carton
\\
Institut de Recherche en Informatique Fondamentale
\\
Universit\'e Paris Diderot
\\
Olivier.Carton@irif.fr
\end{minipage}
}
\end{document}